\documentclass{elsarticle}
\usepackage{xspace,colortbl} 
\usepackage[latin1]{inputenc}
\usepackage[T1]{fontenc} 
\usepackage{graphicx} 
\usepackage{graphics}
\usepackage{verbatim}
\usepackage{bookman} 
\usepackage{amsmath}
\usepackage{amsfonts} 
\usepackage{amssymb,latexsym}
\usepackage{amsthm} 
\usepackage{amscd} 
\usepackage{mathrsfs}
\usepackage{pifont} 
\usepackage{aeguill}
\usepackage{color}

\def\sgn{\mathop{\rm sgn}\nolimits}

\newcommand{\e}[1]{\begin{equation}#1\end{equation}}
\newcommand{\enn}[1]{\begin{equation*}#1\end{equation*}}
\newcommand{\tab}[2]{\begin{array}{#1}#2\end{array}}

\newcommand{\sysnn}[2]{\enn{\left\{\tab{#1}{#2} \right. }}

\newcommand{\K}{\mathcal K^{\theta,\nu}}

\newcommand{\Dx}{\mathcal D_x^\theta}
\newcommand{\Dz}{\mathcal D_z^\theta}

\newtheorem{lem}{Lemma}[section]
\newtheorem{thm}{Theorem}[section]
\newtheorem{prop}{Proposition}[section]
\newdefinition{rmk}{Remark}

\numberwithin{equation}{section}

\begin{document}

\title{Hybrid resonance of 
 Maxwell's equations
in slab geometry\footnote{The 
collaboration leading to this article was started during a visit of Ricardo Weder to INRIA  Paris-Rocquencourt. Ricardo Weder thanks Patrick Joly for his kind hospitality. 
This research  was partially supported  by Consejo Nacional de Ciencia y
Tecnolog\'{\i}a (CONACYT) under project CB2008-99100-F. 
Lise-Marie Imbert-G\'erard 
and Bruno Despr\'es acknowledge 
 the support
of ANR under contract ANR-12-BS01-0006-01. Moreover 
this work was carried out within the framework of the European Fusion Development Agreement and 
the French Research Federation for Fusion Studies. 
It is supported by the European Communities under the contract of Association between Euratom and CEA. 
The views and opinions expressed herein do not necessarily reflect those of the European Commission.
}
}

\author[ljll]{Bruno Despr\'es}
\ead{despres@ann.jussieu.fr}

\author[ljll]{Lise-Marie Imbert-G\'erard\corref{cor1}}
\ead{imbert@ann.jussieu.fr}

\author[rw]{Ricardo Weder}
\ead{weder@unam.mx}

\cortext[cor1]{Corresponding author}

\address[ljll]{Laboratory Jacques Louis Lions, University
Pierre et Marie Curie, 
Bo\^{\i}te courrier 187, 
75252 Paris Cedex 05, France. }

\address[rw]{Departamento de F\'{\i}sica  Matem\'atica, 
Instituto de Investigaciones en Matem\'aticas 
Aplicadas y en Sistemas,
Universidad Nacional Aut\'onoma de M\'exico, Apartado Postal
 20-726, DF 01000, M\'exico. }

\begin{abstract}
Hybrid resonance is a physical
mechanism for the heating of a magnetic plasma.
In our context   hybrid resonance is
a solution of the 
time harmonic Maxwell's equations with smooth coefficients,
where the dielectric tensor is a non diagonal hermitian matrix.
The main  part of this work is dedicated to the  
construction  and analysis of  
a mathematical solution 
of the  hybrid resonance with the limit absorption principle. 
We prove that the limit solution is singular: it  consists
of a  Dirac mass at the origin plus a principal value and a smooth
square integrable function. The formula  obtained for the plasma
heating
is directly related to the singularity.
\end{abstract}

\maketitle

\section{Introduction}

It is known 
in  plasma physics 
that Maxwell's equation in the context
of a strong background magnetic field
may develop  singular solutions even for 
smooth coefficients. This is related to what is called
the  hybrid
resonance \cite{chen,freidberg,branbila} for which we know
no mathematical analysis. Hybrid resonance shows
up in reflectometry experiments \cite{heuraux,gusakov} and heating devices
in fusion plasma \cite{dumont}. 
The energy deposit is resonant and  may exceed by far
the energy exchange which occurs in Landau damping \cite{freidberg,villani}.
The starting point of the analysis  is
from the linearization of Vlasov-Maxwell's equations
of a non homogeneous plasma
around bulk magnetic field $\mathbf B_0\neq 0$.
It yields the non stationary Maxwell's equations with
a linear current 
\begin{equation} \label{eq:cp5}
\left\{
\begin{array}{ll}
-\frac1{c^2}\partial_t \mathbf{E}+
\nabla \wedge \mathbf{B }=\mu_0 \mathbf{J}, &
\mathbf{J}= -e N_e \mathbf{u}_e, \\
\partial_t \mathbf{B}+
\nabla \wedge \mathbf{E }=0, \\
m_e  \partial_t \mathbf{u}_e=
-e\left(\mathbf{E}+ \mathbf{u}_e \wedge \mathbf{B}_0   \right)-
m_e \nu  \mathbf{u}_e. 
\end{array}
\right.
\end{equation}
The electric field is $\mathbf E$ and the magnetic field is 
$\mathbf B$.
The modulus of 
the background magnetic field
$|\mathbf B_0|$ and its direction $\mathbf b_0=\frac{\mathbf B_0}{|\mathbf B_0|} $ will be  assumed  constant in space 
for simplicity in our work.
The absolute value
of the charge of electrons is $e$, the mass of electrons is
$m_e$,
the velocity of light is $c=\sqrt{\frac1{\varepsilon_0 \mu_0}} $ where 
the permittivity of vacuum is $\varepsilon_0$ 
and  the permeability
of vacuum is $\mu_0$.
The third  equation 
 corresponds to moving electrons with velocity $ \mathbf{u}_e$
where the  electronic density $N_e$
is a given function of the space variable.
One implicitly
assumes  an ion bath, which is the reason of the 
friction between the electrons and the ions
with collision frequency $\nu$.
 Much more material about such models can be found in classical
physical textbooks \cite{freidberg,branbila}.
The loss of energy in domain $\Omega$
can easily be computed in the time domain
starting from
(\ref{eq:cp5}).  One obtains
$$
\frac{d}{dt}
\int_\Omega \left(
\frac{\varepsilon_0 \left| \mathbf{E}\right|^2}{2}
+
\frac{ \left| \mathbf{B}\right|^2}{2\mu_0}
+ \frac{m_e N_e  \left| \mathbf{u}_e\right|^2}{2}
\right)=
-\int_\Omega {\nu m_e N_e    \left| \mathbf{u}_e\right|^2
}+
\mbox{ boundary terms}.
$$
Therefore
$
{\cal Q}(\nu)=\int_\Omega {\nu m_e N_e    \left| \mathbf{u}_e\right|^2
}
$
represents the total loss of energy of the electromagnetic field
plus  the electrons in function of the collision frequency $\nu$. Since the 
 energy loss  is necessarily equal
to what is gained by the ions, it will be referred
to as the heating. We will show that in certain conditions
characteristic of the hybrid resonance in frequency domain,
the heating does not vanish for vanishing collision friction.
So a simple 
characterization of   resonant heating can be written as:  ${\cal Q}(0^+)> 0$.
This  apparent  paradox  is the subject of this work.

As we will prove,   the mathematical solution of the time frequency
formulation is  not square integrable. So that, 
hybrid resonance is  a non standard
phenomenon  in the context
of the  mathematical theory of Maxwell's equations for which
we refer to \cite{dautray:lions,cessenat,monk, weder}. 
The situation can be compared with the mathematical theory of 
metametarials. In  \cite{weder2,weder3}  the electric permittivity and 
magnetic permeability 
tensors are  degenerate -i.e. they have zero eigenvalues- in surfaces, but they remain  positive definite. In this case, the solutions are singular, but the problem remains coercive. See also \cite{chen:lipton}. 
In \cite{bonnet1,bonnet2,bordeaux} 
the coefficient changes in a discontinuous way from being positive to negative. In this situation coerciveness is lost, but as the absolute value of the coefficient is bounded below by a positive constant,  the solutions are regular. 
In our case we have both difficulties at the same time.  As  the
 coefficient $\alpha$ (see below) goes from being positive to negative  in 
a continuous way, its absolute value is zero at a point, and, in consequence, 
our problem is not coercive and there are singular solutions.


\subsection{Maxwell's equations  in frequency domain }


We introduce  the notations  needed to 
detail the physics of the problem and to 
formulate our main result.
Writing (\ref{eq:cp5}) in   the frequency domain, that is  
$\partial_t  = -i\omega$, yields
\begin{equation}\label{eq:cp1}
\left\{
\begin{array}{ll}
\frac1{c^2}i\omega  \mathbf{E}+
\nabla \wedge \mathbf{B }=-\mu_0 e  N_e \mathbf{u}_e, \\
-i\omega  \mathbf{B}+
\nabla \wedge \mathbf{E }=0, \\
-im_e  \omega \mathbf{u}_e=
-e\left(\mathbf{E}+ \mathbf{u}_e \wedge \mathbf{B}_0   \right)-
m_e \nu  \mathbf{u}_e. 
\end{array}
\right.
\end{equation}
One 
 computes the velocity using the  third equation 
\begin{equation} \label{eq:cp6}
\widetilde{\omega} \mathbf{u}_e+
\omega_c i  \mathbf{u}_e \wedge \mathbf{b}_0  
= -
\frac{e}{m_e}i\mathbf{E}
\end{equation}
where the cyclotron frequency is $\omega_c= \frac{e|\mathbf{B}_0|}{m_e}$,
$\mathbf{b}_0=\frac{\mathbf{B}_0}{|\mathbf{B}_0|}$ is the normalized
magnetic field and
$
\widetilde{\omega}=\omega +i \nu$ is the equivalent a priori
complex pulsation. 
This is a linear equation. Assuming that $\mathbf{b}_0=(0,0,1)$ one gets 
\begin{equation} \label{eq:cp12}
 \mathbf{u}_e=-\frac{e}{m_e}i
\left(
\begin{array}{ccc}
\frac{\widetilde{\omega}}{\widetilde{\omega}^2-\omega_c^2  } &
-i \frac{\omega_c}{\widetilde{\omega}^2-\omega_c^2  } 
 & 0\\
i \frac{\omega_c}{\widetilde{\omega}^2-\omega_c^2  } 
& \frac{\widetilde{\omega}}{\widetilde{\omega}^2-\omega_c^2  } & 0\\
0 & 0 & \frac1{\widetilde{\omega}}
\end{array}
\right)
\mathbf{E}.
\end{equation}
It is then easy to eliminate $\mathbf{u}_e$ from the first equation
of the system (\ref{eq:cp1})
and to obtain the time harmonic Maxwell's equation
\e{\label{eqEinit}
\nabla \wedge  \nabla \wedge
 \mathbf E - \left( \frac{\omega}{c} \right)^2 
\underline{\underline{\varepsilon}}(\nu)
\mathbf  E = 0,
}
where  $\omega$ is the frequency,
$c$ the velocity of light.
and  the dielectric tensor  is the one of the cold plasma approximation 
\cite{freidberg,chen}
\begin{equation} \label{eq:cp13}
\underline{\underline{\varepsilon}}(\nu)=
\begin{pmatrix}
1-\frac{\widetilde{\omega}\omega_p^2}{
\omega
\left(\widetilde\omega^2-\omega_c^2\right)}&i \frac{\omega_c \omega_p^2}{\omega \left(\widetilde\omega^2-\omega_c^2\right)}&0 \\
-i \frac{\omega_c \omega_p^2}{\omega \left(\widetilde\omega^2-\omega_c^2\right)}
&1-\frac{\widetilde{\omega} \omega_p^2}{
\omega\left( \widetilde\omega^2-\omega_c^2\right) }&0 \\
0&0& 1 - \frac{\omega_p^2}{\omega \widetilde{\omega}}
\end{pmatrix}
.
\end{equation}
The parameters of the dielectric tensor are
the cyclotron frequency 
$
\omega_c= \frac{e   |\mathbf B_0|}{m_e} 
$
 and 
the plasma frequency
$
\omega_p=\sqrt{\frac{e^2 N_e}{\varepsilon_0m_e}  }
$
 which depends on the electronic density $N_e$. 
We consider in this work $\omega\neq \omega_c$,
that is the frequency is away from the cyclotron frequency,
so that the dielectric tensor is a smooth bounded matrix in our work.
Considering (\ref{eq:cp6}) the heating 
is 
$$
{\cal Q}(\nu)=
\int_\Omega {\nu m_e N_e    \left| \mathbf{u}_e\right|^2
}=
- \mbox{ Re}
\left(
\int_\Omega  e N_e
\left(\mathbf{E}, \mathbf{u}_e   \right)
\right).
$$ 
One can eliminate the electron velocity 
in function of the electric field using 
(\ref{eq:cp12}) rewritten as
$\mathbf u_e  = -\frac{e}{m_e}i \frac{\omega  }{\omega_p^2} \left(
\mathbf I-\underline{\underline{\varepsilon}}(\nu) 
\right)\mathbf E$.
 Therefore
a third formula is  
\begin{equation} \label{eq:Qphys}
{\cal Q}(\nu)= 
{\omega \varepsilon_0 }
\mbox{ Im}\left(
\int_\Omega  
\left(\mathbf{E}, \underline{\underline{\varepsilon}}(\nu) \mathbf{E}   \right)
\right)
\end{equation}
where $ \underline{\underline{\varepsilon}}(\nu)$ is the dielectric tensor
(\ref{eq:cp13}).

A discussion of the heating in the limit
of small collision frequency 
 establishes
 the physical basis of the limit absorption
principle that will be used in this work. 
 Indeed physical values in fusion plasmas are such that
the collision frequency is much smaller than the frequency
($\nu << \omega$) which means that some simplifications
can be done in the dielectric tensor, as in \cite{freidberg} page 197.
The limit tensor for $\nu=0$  is
\e{ \label{eq:bd200}
\underline{\underline{\varepsilon}}(0)=
\begin{pmatrix}
1-\frac{\omega_p^2}{\omega^2-\omega_c^2}&i \frac{\omega_c \omega_p^2}{\omega \left(\omega^2-\omega_c^2\right)}&0 \\
-i \frac{\omega_c \omega_p^2}{\omega \left(\omega^2-\omega_c^2\right)}&1-\frac{\omega_p^2}{\omega^2-\omega_c^2}&0 \\
0&0& 1 - \frac{\omega_p^2}{\omega^2}
\end{pmatrix}
.
}
We notice that $\underline{\underline{\varepsilon}}(0)
=\underline{\underline{\varepsilon}}(0)^*$ is an hermitian matrix, so 
$\underline{\underline{\varepsilon}}(0)$
cannot be used alone to obtain a consistent 
evaluation of the heating.
Linearization of the dielectric tensor yields 
$\underline{\underline{\varepsilon}}(\nu)=
\underline{\underline{\varepsilon}}(0)+
\nu \underline{\underline{\varepsilon}}'(0)+O(\nu^2)$ with
$
\underline{\underline{\varepsilon}}'(0)=i
\left(
\begin{array}{ccc}
\lambda_1 &-i \lambda_2 & 0 \\
i \lambda_2 &  \lambda_1 &  0\\
0 & 0 & \lambda_3 
\end{array}
\right)
$
where 
$
\lambda_1=\frac{\omega_p^2\left(\omega^2 + \omega_c^2  \right)  }{
\omega \left(\omega^2-\omega_c^2   \right)^2  }$,
$\lambda_2=\frac{2 \omega_c \omega_p^2   }{\left(\omega^2-\omega_c^2  
 \right)^2   }
$
and 
$\lambda_3
=\frac{\omega_p^2}{\omega^3}$. Since $\lambda_1^2\geq \lambda_2^2$, one gets that
$-i \underline{\underline{\varepsilon}}'(0)$ is a symetric non 
negative 
 matrix. This correction term
is the one that generates
the heating in (\ref{eq:Qphys}).
 In the sequel we will consider the
 simplified linear approximation 
\begin{equation} \label{eq:espsimp}
\underline{\underline{\varepsilon}}(\nu)=
\underline{\underline{\varepsilon}}(0)+
i\nu \mathbf I,
\end{equation}
 yielding the physical basis of the limit absorption principle.


 \subsection{X-mode equations  in slab geometry}




The hybrid resonance concerns  more specifically 
the $2\times 2$ upper-left
block in (\ref{eq:bd200}),
which corresponds to the transverse electric (TE)
mode, $E=(E_x,E_y,0)$, and $E_x,E_y$, independent of  $z$.
In the limit case $\nu=0$, one gets the system
\e{\label{sysinit}
\left\{\tab{rrrr}{
W &+\partial_y E_x &- \partial_x E_y & =0,
\\  \partial_y W &-\alpha  E_x &-
i\delta 
 E_y&  =0,
\\ -\partial_x W &+ i\delta
E_x 
&-\alpha 
 E_y &=0,
}\right. 
}
where  $W$ and the magnetic field $B_z$ are proportional.
 The coefficients are 
\begin{equation} \label{eq:ad}
\alpha= \frac{\omega^2}{c^2} \left(
1-\frac{\omega_p^2}{\omega^2-\omega_c^2}
\right)
\qquad
\delta=
 \frac{\omega^2}{c^2} \times 
\frac{\omega_c \omega_p^2}{\omega \left(\omega^2-\omega_c^2\right)}.
\end{equation}
Simplified coefficients in slab geometry  will be defined below. 

\begin{figure}[ht]
\begin{center} 
\begin{tabular}{ccc}
\includegraphics[width=8.5cm,angle=0]{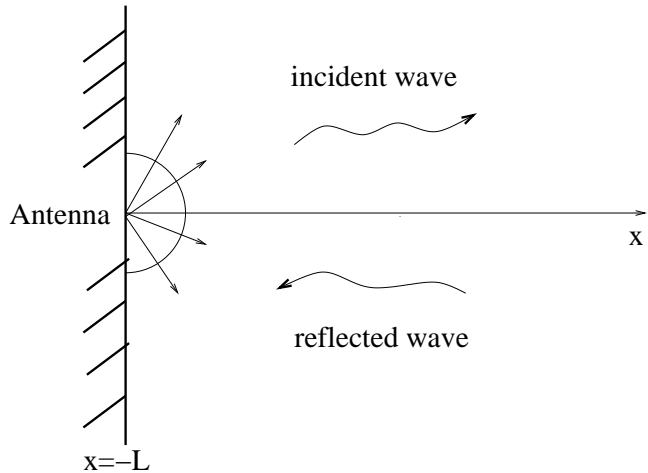} 
\end{tabular}
\end{center}
\caption{X-mode in slab geometry: the domain.
In a real physical device an antenna is on the wall on the
 left and sends an incident  electromagnetic wave
through a medium which is assumed infinite for simplicity.
The incident wave generates a reflected wave.
We will characterize the antenna by the knowledge of the 
non homogeneous boundary condition (\ref{eq:bc}).
The medium is filled with a plasma with dielectric
tensor given by (\ref{eq:bd200}).}
\label{figure:0}
\end{figure}
\begin{figure}[ht]
\begin{center} 
\begin{tabular}{ccc}
\includegraphics[width=8.5cm,angle=0]{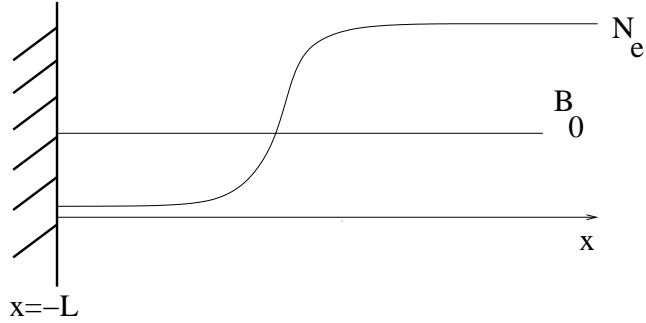} 
\end{tabular}
\end{center}
\caption{X-mode equations in slab geometry: the physical 
parameters. The electronic density $x\mapsto N_e(x)$ is low at the boundary,
and increases towards a plateau.
The background magnetic field $B_0$ is taken as constant for simplicity.}
\label{figure:1}
\end{figure}

In the plasma community this system is referred to as the X-mode
equations, where the letter X  stands for eXtraordinary
mode or eXtraordinary waves. We suspect the reason is the non
 standard behavior of the solutions of this system.
The case where $\omega=  \omega_c$, i.e., when the frequency of the incident wave, $\omega$, is equal to  the cyclotron frequency, $\omega_c$,   will not
be considered in this work. That is we consider that
$\omega\neq \omega_c$. If $\omega<\omega_c$ it is called a low
hybrid resonance. The other case 
$\omega>\omega_c$
is denoted as the upper hybrid resonance.
On the other hand we will assume that
the diagonal coefficient $\alpha$ is smooth and vanishes at $x=0$.
This configuration corresponds to  the  hybrid resonance.

To be more specific we  consider the  simplified 2D domain 
\enn{
\Omega=\left\{(x,y)\in \mathbb R^2, \quad
-L\leq x, \quad
y\in \mathbb{R}, \quad L >0\right\}.
}

Boundary conditions for the Maxwell's equations 
can be of usual types, that is metallic condition
$ n\wedge E =0$,  non homogeneous 
  absorbing boundary condition like
$
curl E + i \lambda n\wedge E =g$ on some parts of the boundary
or even natural absorbing boundary
condition at infinity. %
Concerning  the X-mode equations \eqref{sysinit} 
we consider  a non homogeneous 
 boundary condition 
\begin{equation} \label{eq:bc}
 W+ i\lambda n_x E_y=g \mbox{ on the left boundary }x=-L,
\qquad \lambda>0,
\end{equation}
which models a given source, typically a radiating antenna.
In real Tokamaks this antenna is used to heat  or 
to probe the plasma. 
Such devices are actually being studied
for the purposes of reflectometry and heating of 
magnetic fusion plasmas in the context of the international ITER project:
the ITER project is about the design of  new Tokamak
with enhanced 
fusion capabilities \cite{iter}.

\subsection{Coefficients in slab geometry}

We  consider slab geometry.
That is all   coefficients $\alpha$ and $\delta$ are functions only of the
variable $x$:
$
\partial_y \alpha=\partial_y \delta=0$. 
The main physical hypothesis is that
the extra-diagonal part of the dielectric tensor
is dominant at a finite number of  points, that is 
$$
\alpha(x_i)=0, \; \alpha'(x_i)\neq 0 \mbox{ and }\delta(x_i)\neq 0, \quad 
x_i\in \mathbb R, \; i=1, \dots, N.
$$
To fix the notations we add
other mathematical assumptions which are
reasonable in  the physical context
of idealized reflectometry or heating devices.
We suppose that $N=1$ and $x_1=0$.
 We   will use 
\begin{equation} \label{eq:bd1}
\delta \in {\cal C}^1[-L, \infty[, \qquad
\delta(0) \neq 0 ,
\end{equation} 
\begin{equation} \tag{H1}\label{H1}
\alpha \in {\cal C}^2[-L, \infty[, \qquad
\alpha(0)=0, \quad \alpha'(0)<0. 
\end{equation} 
Moreover 
\begin{equation}\tag{H2} \label{H2}
 \alpha_-\leq \alpha(x) \leq \alpha_+, \; \forall x\in [-L,\infty[,
\qquad
\mbox{ and }\;\;
0<r\leq 
\left| \frac{\alpha(x)}x
\right|
, \; \forall x\in [-L,H]
\end{equation}
where $H>0$.
\begin{figure}[ht]
\begin{center} 
\begin{tabular}{ccc}
\includegraphics[width=8.5cm,angle=0]{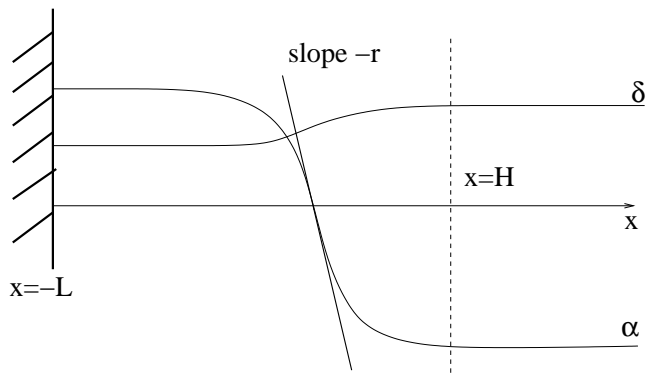} 
\end{tabular}
\end{center}
\caption{X-mode equations in slab geometry:  
parameters of the dielectric tensor deduced from the value
of the physical parameters described in figure \ref{figure:1}, assuming
that $\omega>\omega_c$.
The coefficient $\alpha$ 
decreases from positive to negative values.
It crosses the axis with a slope bounded from below by
$r$.
The coefficient $\delta$ is positive and bounded.
Since the electromagnetic wave is strongly absorbed for $x\geq H$,
we simplify 
by taking all  coefficients constant
for $x \geq H$ because it does not change the physics of the problem. 
}
\label{figure:2}
\end{figure}
We will also assume that the coefficients
are constant at large scale: there exists   $\delta_\infty$ and
$\alpha_\infty$  so that 
\begin{equation} \tag{H3}\label{H3}
\delta(x)=\delta_\infty \mbox{ and }\alpha(x)=\alpha_\infty\qquad 
H\leq x < \infty.
\end{equation}
Therefore $\alpha,\delta \in L^\infty(-L,\infty)$.
We also  assume the problem is coercive at infinity,
\begin{equation} \tag{H4}\label{H4}
\alpha_\infty^2-\delta^2_\infty>0.
\end{equation}
An additional condition is defined by
\e{\tag{H5}\label{H5}
4\|\delta\|_\infty^2 H < r.
}
It expresses the fact that the length of the transition 
zone between $x=0$ and $x=H$ is small with respect to the
other parameters of the problem.
One can refer to Figures \ref{figure:1} and \ref{figure:2}
for a graphical representation. This hypothesis is physically very
reasonable. 
 It is known in the physical community that  this problem
may  be highly singular
at the origin.
With these hypotheses, one can consider as well other coefficients 
 are now normalized $\omega=c=\varepsilon_0=\mu_0=1$.
As explained previously in (\ref{eq:espsimp}),
the solutions in the context of the limit absorption principle 
correspond to adding a complex part to the diagonal coefficient
$\alpha$, that is $\alpha$ is replaced by $\alpha+i \nu$.

\subsection{Main result}

Our main result can be summarized as follows. 
Following the 
 convention introduced  in $(7.1.3)$ from \cite{horm},
we denote by $\widehat{g}$ the Fourier transform of $g$,
$$
\widehat{g}(\theta):=  \int_{\mathbb R}\,g(y)\, e^{-i\theta y}\, dy.
$$
 We need the uniform  transversality assumption  \eqref{H6} 
which is a generalization of assumption  \eqref{H5}. See Section \ref{sec:main}.

\begin{thm} \label{th:main}
Assuming (\ref{H1}-\ref{H6}) and  $ g \in L^2(\mathbb R)$ with $ \widehat{g}$ of compact support, 
there exists a    solution  
of (\ref{sysinit}) with boundary condition (\ref{eq:bc}) that goes to zero at infinity.
This solution is in the sense of distributions and is constructed
with the limit absorption principle by taking the limit $\nu=0^+$ 
in (\ref{eq:fu10}).

A representation formula is 
\begin{equation} \label{eq:fu11}
\left(
\begin{array}{c}
E_x^+ \\
E_y^+ \\
W^+
\end{array}
\right)
(x,y)=\frac1{2\pi}
\int_{\mathbb{R}}
\frac{ \widehat{g}(\theta)}{   \tau^{\theta,+}}
\left(
\begin{array}{c}
P.V. \frac1{\alpha(x)} + \frac{i \pi}{\alpha'(0)} \delta_D
+ u_2^{\theta,+} \\
v_2^{\theta,+} \\
 w_2^{\theta,+}
\end{array}
\right)
 e^{i\theta y  }d\theta.
\end{equation} 
This formula 
depends on a certain transfer coefficient
$ \tau^{\theta,+}$ defined in (\ref{eq:fu2b}), and on three $L^2$ functions
$(u_2^{\theta,+},v_2^{\theta,+},w_2^{\theta,+})$ defined in 
Theorem  \ref{thm5.1}.
 Unless the source term  $g$  is identically zero,  the electric field $E_x$
does not belong to $L^1_{ {\rm loc}}\left((-L,\infty)\times \mathbb R\right)$.
 The other
 components  are always
more regular: in particular
  $E_y^+,W^+\in L^2\left((-L,\infty)\times  \mathbb R\right)$.

The value of the resonant 
heating is
\begin{equation} \label{eq:limheat}
{\cal Q}^+ =
\frac1{2}
\int_{\mathbb{R}}  \frac{  \left|
\widehat{g}(\theta) \right|^2 }{ \left|
\alpha'(0)
\right|
 \left|  \tau^{\theta,+} \right|^2 }
d\theta>0.
\end{equation}
\end{thm}


\begin{rmk}
An essential consequence of this analysis is
 the resonant  heating ${\cal Q}^+$ which
is directly related to the singularity $
P.V. \frac1{\alpha(x)} + \frac{i \pi}{\alpha'(0)} \delta_D
$ of the mathematical solution $E_x$.
The singularity is not an artifact
of the model. It is on the contrary a direct way to measure
the amount of heating provided to the ions by the electromagnetic wave.
Concerning $E_y$ and $B_z$ which are integrable, a logarithmic
divergence is still present in the solution as
seen in the solution (\ref{eq:ww}) of the Budden problem, or also
in (\ref{eq:vvtilde}-\ref{eq:wwtilde}) for example.
\end{rmk}

\begin{rmk}
The hypothesis $\delta(0)\neq 0$ is technically important
in our work. It is used  two times:
in the solution of the Budden problem (\ref{eq:whit})
 and in the normalization (\ref{eq:delta0})
of the singular
solution since we divide by $\delta(0)$.
\end{rmk}

To our knowledge this is the first time that such
 formulas are written where all terms are explicitly given.
A similar but much less precise formula can be found
in \cite{chen} derived by means of analogies, see also \cite{piliya}.
The  formulas (\ref{eq:fu11}-\ref{eq:limheat})
 have  been 
confirmed by numerical simulations \cite{imb:these}
where additional information may be found about the 
the case where (\ref{H5},\ref{H6})
are not satisfied. It must be mentioned that the numerical tests
show a fast pointwise convergence
of the numerical
solution to the exact one, except at the origin of course. 
Moreover our numerical tests
show that a large part  of the incoming energy of the wave may be  
absorbed by the heating, around 90\% in some cases.
This is for example the case for the Fourier mode
$\theta=0$
with $L=2$: the physical coefficients in
(\ref{eq:ad}) are $c=\omega=2$ and $\omega_c=1$,
so that $\alpha=1-2\delta$. We consider the profiles
$$
\alpha(x)=\left\{
\begin{array}{lc}
1 ,& -L\leq x \leq -1, \\
-x ,& -1 \leq x \leq 3, \\
-3 ,& 3\leq x<\infty,
\end{array}
\right.
\mbox{ and }
\delta(x)=\left\{
\begin{array}{lc}
0 ,& -L\leq x \leq -1, \\
\frac{x+1}2 ,& -1 \leq x \leq 3, \\
2 ,& 3\leq x<\infty,
\end{array}
\right.
$$
which satisfy additionally $|\alpha_\infty|>|\delta_\infty|$
and the fact that the electronic density is increasing from the left to the
 right.
For the calculation of the heating we use equation 
(\ref{eq:dieze}) with $M=-L$ , $N=\infty$, and 
 we observe that for
normal incidence $W_2^{0,0}=\frac{d}{dx}V_2^{0,0}$ and that for
$-L\leq x \leq -1$, $V_2^{0,0}$
is a
linear combination of an incoming plane wave and a reflected plane wave.
Furthermore, we compute numerically the singular solution  
$\mathbf U_2^{0,\nu}$
taking $\nu=10^{-3}$
as a small regularization parameter.
The efficiency of the heating is defined as the ratio
of the heating $\cal Q$ over the incoming energy.
In this case our calculations show an efficieny of around $95\%$.
Another calculation in  oblique incidence $\theta=\cos \frac\pi4$
shows an efficiency still around $76,7\%$.
These values indicate a high efficiency.

The method of the proof is based on an original singular 
integral equation  attached to the Fourier solution. Introduced
in  the seminal work of Hilbert \cite{hilbert}
and Picard \cite{picard},
this type of integral equation is referred to as 
integral equation of the third kind, by comparison with
the more classical equations of the first and second kind.
Some references about this type of equations
may be found in 
\cite{bart:warnock,shulaia} for mathematical analysis,
and \cite{vankampen,case,frye} for relation with theory of particles
or plasma physics. Our  results 
are therefore reminiscent of those 
of Bart and Warnock \cite{bart:warnock}, even if our kernel
does not satisfy exactly their hypothesis since it is less regular:
that is  the solution is the sum of a Dirac mass plus a principal
value (plus a regular part).
In their work it is stressed that non uniqueness is the rule for such
equations.
 In our case,  we are able to  obtain uniqueness
by means of the limit absorption principle which is a physically based
selection principle.
One originality of this work is the analysis of the 
properties of this singular equation for which we found no equivalent
in the classical literature \cite{ab:ste,bateman,bouche}.
The result  will be obtained 
with the limit absorption  principle
combined with a specific original integral representation
of the solution.
The loss of regularity of the electric field 
 is counter intuitive with respect to the standard theory
of existence and uniqueness for solutions
of time harmonic Maxwell's equations  \cite{dautray:lions,cessenat,monk, weder}.
 The essential part of the proof consists in showing
that the Fourier transform  $\widehat{E}_x$ may be composed 
of three contributions:
a Dirac mass  at $x=0$; a non integrable function proportional
to $\frac1{\alpha(x)}$, 
that is interpreted  as a distribution in the sense of principal value; and a regular part. 
 The condition $ \alpha'(0) <0$ guarantees
that the coefficient in front of the Dirac mass is finite. Moreover, the condition (\ref{H5})
 simplifies some parts of the mathematical analysis.
The solution is a priori non unique since
the limit absorption  principle generates two solutions
depending on the sign
of the regularization.
The heating of the plasma (\ref{eq:limheat})
is directly related to the singular part of the solution.

\subsection{Organization}

This work is organized as follows.
Section \ref{sec:2} is devoted to basic considerations.
In the next section we introduce a regularization parameter, and we propose a specific
integral representation of the solution.
After that we recall the Plemelj-Privalov theorem and explain 
why it cannot be used directly for our problem.
Section \ref{sec:5} is where we prove the properties
of the solutions 
of the regularized equations.
In particular, we show that one basis function
has a fundamental singularity.
Next, in Section \ref{sec:6} we define the limit spaces.
The main theorem is finally proved in section \ref{sec:main}.

\section{Basic considerations} \label{sec:2}

In this section we rederive the phase velocity, compute
the analytic solutions of the simplified
Budden problem and introduce the limit absorption principle.

\subsection{Phase velocity}

Recall that the phase velocity measures the velocity of individual
Fourier modes.

\subsubsection{Constant coefficients}

Let us consider first that $\alpha$ and $\delta$ are constant
at least locally.
A plane wave
$(E_x,E_y)=Re^{i (k_1 x +k_2 y)  }$, $R\in \mathbb {C}^2$, 
is solution of X-mode equations (\ref{sysinit})
if and only if
$$
\left[\left(
\begin{array}{cc}
k_2^2& -k_1\,k_2 \\
-k_1\,k_2 & k_1^2 
\end{array}
\right) -\frac{\omega^2}{c^2}
\left(
\begin{array}{cc}
\alpha & i \delta \\
-i \delta & \alpha 
\end{array}
\right)\right]
 R=0, \qquad k=(k_1,k_2)\in \mathbb{R}^2.
$$
We assume that $c=1$ for simplicity.
We set $k=|k| d$ with $d=(\cos \theta, \sin \theta)$ 
the direction of the wave.
The phase velocity
$
v_\varphi=\frac{\omega}{|k|}
$
is solution of the eigenvalue problem
$$
\left(   
\begin{array}{cc}
\sin  \theta^2 - v_\varphi^2 \alpha & - \cos \theta \sin \theta -
i  v_\varphi^2 \delta \\
- \cos \theta \sin \theta +
i  v_\varphi^2 \delta &
\cos \theta^2 - v_\varphi^2 \alpha 
\end{array}
\right)
R=0.
$$
The determinant of the matrix is
$
D=v_\varphi^4\left(\alpha^2-\delta^2   \right)
- v_\varphi^2\alpha$. 
Setting $ D=0$ we obtain  the phase velocity:
$
v_\varphi^2=\frac{\alpha  }{\alpha^2-\delta^2  }$. 

\subsubsection{Non constant coefficients}

Let us assume for example that $\alpha=-x$
and that $\delta=1$ which is locally
compatible
with the general assumptions
of Figure \ref{figure:2}.
We  plot in Figure \ref{fig1} the phase velocity as a function of the 
horizontal space coordinate. 
When the phase velocity is real  we are in a propagating region, and when the phase velocity  is pure imaginary  we are in a non-propagating region. 
One distinguishes
two cutoffs where the local phase velocity is infinite
$$
\mbox{Cutoff}: \quad \alpha(x)=\pm \delta(x)
$$
and one resonance where the phase velocity
is null
$$
\mbox{Resonance}: \alpha(x)=0.
$$
This structure is characteristic of the  hybrid resonance.

\begin{figure}[ht!]
\begin{center} 
\includegraphics[width=10cm,angle=0]{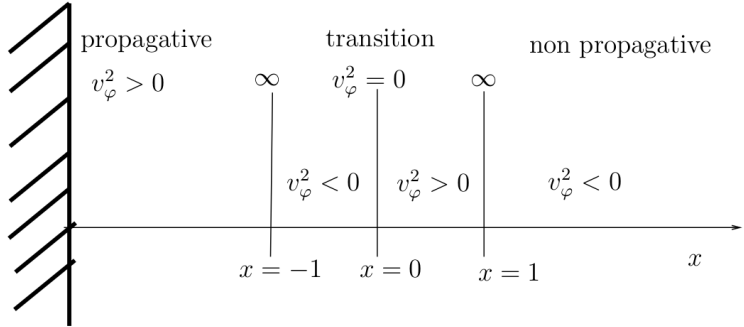} 
\end{center}
\caption{Sign of the square of the phase velocity
$v_\varphi^2=\frac{x}{1-x^2}$, for $\alpha=-x$ and $\delta =1$.} 
\label{fig1}
\end{figure}
\begin{rmk}{\rm
In what follows we always take $ \omega=c=\varepsilon_0=\mu_0=1$.}
\end{rmk}
\subsection{The Budden problem} 
In the case where the solution is independent of $y$, what for the plane waves corresponds to normal  incidence, that is $\theta=0$,
 the system (\ref{sysinit})
  is called
the Budden problem \cite{chen}
$$
\left\{\tab{lr}{
W-  E_y'  &  =0, \\
  -\alpha  E_x -i\delta  E_y & =0, \\
 - W'  + i\delta  E_x -\alpha E_y & =0.
}\right.
$$
After elimination of $E_x$ and $W$ we obtain that,
$$
-E_y''+\left(\frac{\delta^2}{\alpha}    - \alpha \right)E_y=0
$$
This equation  can be solved  analytically in some cases which helps
a lot to understand the singularity of the general problem.
%
Let us consider that $\alpha=-x$ and  $\delta$ is solution of 
$
\frac{\delta^2}{x}    -x=-\frac14+\frac1x$. 
The positive solution is
$
\delta(x)=\sqrt{x^2-\frac{x}4+1    }>0$. 
The y-component of the electric field is solution
of
\begin{equation} \label{eq:whit}
E_y''+
\left(-\frac14+\frac1x   \right)E_y=0.
\end{equation}
This equation is of Whittaker type  \cite{ab:ste,bateman}.
It is a particular case of the confluent hypergeometric
equation, and can also be rewritten under the Kummer form.
The general theory shows that the first fundamental  solution is regular
$$ 
v(x)=e^{-\frac{x}2}x
$$
Indeed
$v'(x)=e^{-\frac{x}2}\left(1-\frac{x}2   \right)$ and 
$v''(x)=e^{-\frac{x}2}\left(-1+\frac{x}4   \right)$, 
so that
$
v''+
\left(-\frac14+\frac1x   \right)v=0$. 
Let us consider a second solution $w$ with linear independence with respect
to the first one. The linear independence can be characterized
by the normalized Wronskian relation
$
v(x)w'(x)-v'(x)w(x)=1$. 
Seeking for a representation $w=vz$, one gets that
$$
v^2 z'=1 \Rightarrow z= \int \frac{dx}{v^2}\Rightarrow 
w=v \int \frac{dx}{v^2}= x \,e^{-x/2}\, \int \frac{e^x}{x^2}.
$$
Moreover, from formulas $8.212$ of \cite{GR},
$$
 \int \frac{e^x}{x^2}= - \frac{e^x}{x}+  \int \frac{e^x}{x}=  - \frac{e^x}{x}+ E_i(x),
$$
where $E_i(x)$ is the Exponential-integral function. It follows that
$
w(x)=- e^{x/2}+ x\, e^{-x/2}\, E_i(x)$. 
Furthermore from formulas $8.214$ of \cite{GR}
$$
E_i(x)=C+ \ln|x|+\sum_{j=1}^\infty \frac{x^j}{j\cdot j!}.
$$
It follows that,
\begin{equation} \label{eq:ww}
w(x)=-1 +Cx+x \, \ln|x|+O(|x|), \quad |x| \rightarrow 0.
\end{equation}
We notice that the second function $w$ is bounded, but
non regular at origin. It shows the subtleties
associated with the singular  Whittaker equation  (\ref{eq:whit}).
Nevertheless we note that
 the general form of the $y$ component of the electric field
of the Budden problem is bounded
$$
E_y=a v + b w \Rightarrow  E_y \in L^\infty(]-\epsilon, \epsilon[).
$$
The $x$ component of the electric field is  more singular.
It is a linear combination of
two functions, the first one which is regular and bounded
$$
E_x^v(x)=i\frac{\sqrt{x^2-\frac{x}4+1    }  }{x}v(x)=
i  e^{-\frac{x}2} \sqrt{x^2-\frac{x}4+1    },
$$
and the second one 
which is singular at origin since $w(0)=-1$
$$
E_x^w(x)=i\frac{\sqrt{x^2-\frac{x}4+1    }  }{x}w(x).
$$
The general form of the $x$ component
of the electric field is a linear combination
of these two functions.
Since $E_x^w\not \in L^2(]-\epsilon, \epsilon[)$, we notice
that the electric field is not a square integrable
function in general. 

\subsection{Limit absorption principle}

We will develop a regularized approach
to give a rigorous meaning to the solution at all incidences.
This regularized approach is based on the limit absorption principle.
One considers a parameter $\nu\neq 0$ (the precise sign will be justified
later) and the regularized 
problem 
with unknown 
 $(E_x^\nu,E_y^\nu,W^\nu)$ 
\e{\label{sysmu}
\left\{
\begin{array}{cccr}
W^\nu & +  \partial_y E_x^\nu  & - \partial_x E_y^\nu  & =0,
\\  
\partial_y W^\nu  & -(\alpha(x)+i\nu) E_x^\nu &
-i\delta (x) E_y^\nu & =0,
\\ -\partial_x W^\nu &
+ i\delta (x) E_x^\nu &
-(\alpha(x)+i\nu) E_y^\nu & =0.
\end{array}
\right.
}
 The  regularization parameter $\nu$ can be interpreted as a small collision frequency. 

A further simplification consists in Fourier reduction.
Since the coefficients do not depend on the $y$ variable, one can perform 
the usual  one dimension reduction.
The system that will be studied in this article is obtained by applying the Fourier transform 
to the regularized system \eqref{sysmu}. Denoting the unknowns $(U,V,W)$ it yields
\begin{equation}\label{sys0:hatmu}
\left\{
\begin{array}{cccr}
  W & 
+  i\theta U & - V'& =0, \\
  i\theta W & -(\alpha (x)+i\nu) U 
&-i\delta (x) V & =0,\\ 
  - W' &
+ i\delta (x) U &
-(\alpha(x)+i\nu) V & =0.
\end{array}
\right.
\end{equation} 
Here the notation 
 $'$ denotes the derivative with respect to the $x$ variable.

\section{A general integral representation} \label{sec:3}

We begin by some notations.
Let us denote by $(A_\nu,B_\nu)$ the 
two fundamental solutions of the modified equation
\begin{equation} \label{eq:bd5}
-u''-(\alpha(x)+i\nu)u=0,
\end{equation}
with the usual normalization
\begin{equation} \label{eq:bd20}
A_\nu(0)=1, \quad
A_\nu'(0)=0 \qquad 
\mbox{ and }
B_\nu(0)=0, \quad B_\nu'(0)=1.
\end{equation}
Various usual continuity estimates of $A_\nu$ and $B_\nu$
can be  derived: we refer for example 
for  the appendix of \cite{imb2}.
Let us denote $\Dz $ the operator $i \theta\partial_z  -i
\delta(z)$ applied to any function $h$, that is
\begin{equation} \label{eq:bd6}
\Dz h = i \theta\partial_z  h - i \delta(z) h.
\end{equation}
Let us define the kernel
\begin{equation}\label{eq:bd9}
k^\nu(x,z)=B_\nu(z)A_\nu(x)-B_\nu(x)A_\nu(z).
\end{equation}
Next we define
\begin{equation} \label{eq:bd15}
K_1^{\theta,\nu}(x,z; G) =
\left\{
\begin{array}{lll}
\displaystyle \frac{\Dx \Dz k^{\nu}(x,z)}{\alpha(x)+i\nu}, & \mbox{ for }
G\leq  z \leq x & \mbox{ or }x\leq z\leq G ,\\
0 ,&   \mbox{ in all other cases}.
\end{array}
\right.
\end{equation}
Let us define the kernel sequence by
\e{ \label{eq:bd8}
 K_{n+1}^{\theta,\nu}(x,z ; G)  =
 \int_{G}^x \frac{\Dx \Dz k^{\nu}(x,t)}{\alpha(x)+i\nu}
 K_{n}^{\theta,\nu}(t,z; G) dt.
}
The sum is 
\begin{equation}\label{eq:ResKer}
 \K (x,z; G) = \sum_{n=0}^\infty K_{n+1}^{\theta,\nu}(x,z; G).
\end{equation}
The integration domain is
 centered on $G$, that is
\begin{equation} \label{eq:bd79}
\mbox{supp}\left(K_1^{\theta,\nu}(\cdot,\cdot; G)\right)
\subset \left\{
(x,y)\in \mathbb{R}^2;\quad 
G\leq  z \leq x  \mbox{ or }x\leq z\leq G
\right\}\equiv {\cal D}_G,
\end{equation}
which yields as well:
$
\mbox{supp}\left(\K(\cdot,\cdot; G)\right)
\subset {\cal D}_G$. 

\begin{prop} \label{prop:bd1}
Any   triplet 
$(U ,V , W )$
solution of  the regularized system \eqref{sys0:hatmu} admits the
following integral representation.

\begin{itemize}
\item One first chooses an arbitrary
 reference point  $G\in [-L,\infty[$.
\item The $x$ component of the 
electric field is solution of the integral equation
 \e{ \label{eq:bdeq:bd13}
U (x) - \int_{G}^x \frac{
\Dx \Dz k^{\nu}(x,z)}{\alpha(x)+i\nu} U (z) dz = \frac{F^{\theta,\nu}(x)}{\alpha(x)+i\nu},
}
where the right hand side is
\begin{equation} \label{eq:bdeq:bd14}
{F^{\theta,\nu}(x)}=
a_G \Dx  A_\nu(x)+b_G \Dx  B_\nu(x) 
\end{equation}
and the kernel is given in (\ref{eq:bd6}-\ref{eq:bd9}).
The solution of this integral equation is naturally
 provided by the resolvent integral formula
\begin{equation} \label{eq:bd10}
U (x)
=
\frac{F^{\theta,\nu}(x)}{\alpha(x)+i\nu} +
 \int_{G}^x \K(x,z; G) 
\frac{F^{\theta,\nu}_G(z)}{\alpha(z)+i\nu} dz
\end{equation}
where the resolvent kernel is constructed in (\ref{eq:ResKer}).

\item The $y$ component of the electric field is recovered  as
\begin{equation} \label{eq:bd11} 
V (x) =
a_G A_\nu(x)+b_G
B_\nu(x) 
 +  \int_{G}^x \Dz k^\nu(x,z)
U (z) 
dz ,
\end{equation}
and the vorticity is recovered as 
\begin{equation} \label{eq:bd12} 
W (x) =
a_G A_\nu'(x)+b_G 
B_\nu'(x) 
 +  \int_{G}^x \partial_x \Dz k^\nu(x,z)
U (z) 
dz .
\end{equation}

\item
The two complex numbers $(a_G,b_G)$ 
solve the linear system
\begin{equation} \label{eq:bd51}
\left\{
\begin{array}{ll}
a_G   A_\nu(G)+ b_G B_\nu(G)=  V(G), \\
a_G   A_\nu'(G)+ b_G B_\nu'(G)= W(G).
\end{array}
\right.
\end{equation}
\end{itemize}
\end{prop}

\begin{proof}
Eliminating $W$ from the first and third 
equations of (\ref{sys0:hatmu}) 
gives
 \enn{
-V ''-(\alpha+i\nu)V = f
\qquad  \mbox{ with }
f=-i\theta U'-i\delta U.
}
 Since the Wronskian is constant,  it follows from the normalization (\ref{eq:bd20}) that
$A_\nu B_\nu' -A_\nu'B_\nu=1$.
  Then,  from the variation of constants formula,
 \begin{equation} \label{eq:bdv}
V(x) = a_f A_\nu(x)+b_f B_\nu(x)+
 \int_{G}^x f(z) k^\nu(x,z)dz, 
\quad 
\forall x .
\end{equation}
where $a_f$ and $b_f$ are two integration constants.
Now we replace  $f$ by the corresponding function of $U$ and perform
 the integration by part
$$
\int_G^x U'(z)  k^\nu(x,z)dz=U(x) k^\nu(x,x) - U(G) k^\nu(x,G)
- \int_G^x U(z)  \partial_z k^\nu(x,z)dz.
$$
Since $ k^\nu(x,x)=0$ there is a simplification. Therefore (\ref{eq:bdv})
 yields 
\eqref{eq:bd11}
with  $a_G = a_f +i\theta U (G)B_\nu(G)$
 and $b_G = b_f -i\theta U (G)A_\nu(G)$.
Next we eliminate $W$ from the first and
second equations of  (\ref{sys0:hatmu}) and obtain
\begin{equation} \label{eq:bdv2}
-i\theta V'-\theta^2 U+ (\alpha+i\nu) U+i\delta V=0.
\end{equation}
The derivative of (\ref{eq:bd11}) yields
$$
V'(x) =
a_G A_\nu'(x)+b_G
B_\nu'(x) 
 +  \int_{G}^x \partial_x \Dz k^\nu(x,z)
U (z) dz + \Dz k^\nu(x,x)U(x).
$$
Since $ \Dz k^\nu(x,x)= i\theta \left(A_\nu B_\nu' -B_\nu A_\nu '    \right)=
i\theta$, one gets the identity
$$
V'(x) =
a_G A_\nu'(x)+b_G
B_\nu'(x) 
 +  \int_{G}^x \partial_x \Dz k^\nu(x,z)
U (z) dz + i\theta  U(x).
$$
Plugging this expression in (\ref{eq:bdv2}) and performing all simplifications
we obtain  the integral equation
 (\ref{eq:bdeq:bd13}).
Finally, we get the last  integral formula (\ref{eq:bd12})
from  $W=-i\theta U+V'$.
The linear system \eqref{eq:bd51} is obvious from (\ref{eq:bd11}-\ref{eq:bd12})
at $x=G$.
\end{proof}

Following  \cite{picard}, the equation 
 (\ref{eq:bdeq:bd13})
is an integral equation of the third kind in the case
$\nu=0$. In this case 
 the theory
is  rather incomplete regarding existence and uniqueness \cite{bart:warnock}. 
 However as long as $\nu \neq 0$, the solution
 based on these integral equations is uniquely defined.
Then, the question is to determine the behavior 
of these solutions when $\nu$ goes to $0$.
Moreover, different choices of $G$ will give different kind of 
information.
A strategy to study of the limit solution $\nu\rightarrow 0$
can be the following:  {\it  Choose  an optimal $G$, so that
a) the integration constants  $(a_G,b_G)$ are 
easy to determine, and  b) the resolvent kernel 
${\cal K}^{\theta,\nu}(\cdot, \cdot;G)$
admits a limit as $\nu\rightarrow 0$.}
Considering the form of the right hand side in (\ref{eq:bd10}), 
a convenient tool is the Plemelj-Privalov Theorem \cite{musk,privalov}.
Unfortunately, we will see that  a fundamental singularity  of the kernel 
${\cal K}^{\theta,\nu}(\cdot, \cdot;G)$  prevents any simple limit procedure. A more convenient
technique will be proposed in Section \ref{sec:5}.

\section{Singularity  of the kernels} \label{sec:4}

A fundamental tool 
in order to pass to the limit in singular integrals
is the Plemelj-Privalov theorem \cite{musk,privalov}. 
However, to apply this theorem to pass to the limit $\nu \rightarrow 0$ in   equation (\ref {eq:bd10}) it is necessary that the kernel $\K(x,z)$ be a H\"older continuous function of $z$ for each fixed $x$. Unfortunately, this regularity is not available in our case.
To illustrate this phenomenon, we  study only 
the first term of the series (\ref{eq:bd8})
that defines $\K$, namely
\begin{equation} \label{eq:23}
\K_1(x,z):=\frac{  \Dx \Dz k^\nu(x,z)}{\alpha(x)+i\nu}.
\end{equation}
We consider two cases.

\subsection{First case: $G\neq 0$}

In this case there exists $(0,z)\in \mathcal D_G$ with $z \neq 0$. In the limit case  $\nu =0$ one has that $\mathcal K^{\theta,0}_1(x,z)$ admits the local expansion:

$$
\mathcal K^{\theta,0}_1(x,z) \approx \frac{1}{x \,\,\alpha'(0)} \,  \mathcal D^\theta_x \, \mathcal D^{\theta}_z \,  k^0(x,z).
 $$
Therefore, $\mathcal K^{\theta,0}_1(x,z)$ blows up as $ x \rightarrow 0$. 

\subsection{Second case: $G=0$}

We turn to the case $G=0$. 
We begin with a preliminary result.

\begin{prop}
One has 
\begin{equation} \label{eq:bd16}
(\Dx \Dz k^\nu)(x,x)=0 \quad \forall x\in \mathbb{R}.
\end{equation} 
\end{prop}
\begin{proof}
Indeed by construction
$$
(\Dx \Dz k^\nu)(x,x)=
-\delta(x) \delta (x) k^\nu(x,x)
$$
$$
+\theta\, \delta (x) \left( (\partial_x  k^\nu)(x,x) +(\partial_z  k^\nu)(x,x)
\right)
- \theta^2(\partial_x \partial_z k^\nu)(x,x).
$$
We notice that by definition
$k_\nu(x,x)=0$ for all $x$ so the first contribution vanishes
in $(\Dx \Dz k^\nu)(x,x)$. One also has that 
$$
 (\partial_x  k^\nu)(x,x) +(\partial_z  k^\nu)(x,x)
$$
$$
= B_\nu(x) A_\nu'(x)-B_\nu'(x) A_\nu(x)+ B'_\nu(x) A_\nu(x)-B_\nu(x)
 A_\nu'(x)=0,
 $$
so, the second contribution vanishes also.
Furthermore,
 $$
 (\partial_x \partial_z k^\nu)(x,x)= B_\nu'(x)\, A_\nu'(x)- B_\nu'(x)\,A_\nu'(x)=0.
 $$
 This completes the proof of equation (\ref{eq:bd16}).
\end{proof}

\begin{prop}\label{prop:limker}  
The limit kernel
$
 \frac{  \Dx \Dz k^{\nu=0}(x,z)}{\alpha(x)}
$
 belongs to  $L^\infty\left({\cal D}_0\right)$.
\end{prop}
\begin{proof}
A first order Taylor expansion  of $\Dx \Dz k^\nu$
around 0 yields
$$
\Dx \Dz k^\nu(x,z)=
 \alpha_\nu x + \beta_\nu z +O(|x|^2+|z|^2).
$$
Notice that (\ref{eq:bd16}) implies $\beta_\nu=-\alpha_\nu$.
The coefficient $ \alpha_\nu$ is easily computed using 
$
(\Dx \Dz k^\nu)(x,0)=\Dx A_\nu(x) \Dz B_\nu(0)- \Dx B_\nu(x) \Dz A_\nu(0)
$
and the definition (\ref{eq:bd5}-\ref{eq:bd20}). One gets that
$
\Dx  A_\nu(x)=-i\delta(0)-i\delta'(0)x +\theta \nu x + O(x^2)
$
and
$
\Dx  B_\nu(x)=i\theta-i\delta(0)x +O(x^2)$.
So
$$
(\Dx \Dz k^\nu)(x,0)=
\left( -i\delta(0) -i\delta'(0)x +\theta \nu x +O(x^2) \right)i\theta
$$
$$
-\left(i\theta  -i\delta (0)x + O(x^2)  \right)  (-i\delta(0) )
$$
$$
=
\left( \delta(0)^2 +\theta \delta'(0)+ i
\theta^2\nu  \right)x + O(x^2).
$$
This coefficient $\alpha_\nu$ being constant, one obtains
that
\begin{equation} \label{eq:bd22}
\varphi_x(z):= \frac{  \Dx \Dz k^{\nu=0}(x,z)}{\alpha(x)}
=\frac{ \left( \delta(0)^2 +\theta \delta'(0)
\right) (x -z) +O(|x|^2+|z|^2) }{
\alpha(x)}.
\end{equation}
This expansion is valid for $(x,z)\in {\cal D}_0$
(the domain ${\cal D }_0$ is defined in (\ref{eq:bd79})):
in this case $|x-z|\leq |x|$ and $|z|\leq |x|$. Moreover, since $ \alpha(x)= x(\alpha'(0) + O(1))$ we obtain that
$
\left| \varphi_x(z)\right| \leq
\frac{  \left| \delta(0)^2 +\theta \delta'(0)\right|
 }{
| \alpha'(0)|}+O(|x|)$. 
Since there is no such difficulty
for $x$ away from 0, this inequality ends the proof of the proposition.
\end{proof}

\begin{rmk}\label{rk:kerb}
A similar property holds for $ \frac{  \Dx \Dz k^{\nu}(x,z)}{\alpha(x)+i\nu}$ which also belongs to  $L^\infty\left({\cal D}_0\right)$ for all $\theta$ and uniformly for $ \nu \in [-1,1] \setminus \{0\}$, that is
\begin{equation}\label{eq:ddfond}
\left\|\frac{  \Dx \Dz k^{\nu}(x,z)}{\alpha(x)+i\nu}\right\|_{L^\infty\left({\cal D}_0\right)} \leq  C^\theta, \quad \nu \in [-1,1] \setminus \{0\}.
\end{equation}
\end{rmk}
Such estimate is  sufficient to control some  $L^\infty$  bounds
of the series that defines  the iterated
kernel $\K(x,z;0)$:
\enn{\tab{rl}{
 \left| K_{n+1}^{\theta,\nu}(x,z ; 0)\right|  &
\displaystyle =
\left| \int_{0}^x \frac{\Dx \Dz k^{\nu}(x,t)}{\alpha(x)+i\nu}
 K_{n}^{\theta,\nu}(t,z; 0) dt\right|,
\\ &\displaystyle \leq C_\theta^{n+1}
\underbrace{ \int_{0<x_1<\dots<x_{n}<x} \prod_{1\leq i \leq n}\ dx_i}_{\frac{x^n}{n!}},
}}
so that
\enn{
\left|\K(x,z;0)\right|  =\left| \sum_{n=0}^\infty  K_{n+1}^{\theta,\nu}(x,z ; 0) \right| \leq C_\theta \left( e^{C_\theta H}-1. \right)
}
However,  $L^\infty$ bounds are
 not sufficient to show
that  $\K(x,z;0)$ is of 
  H\"older class in $z$ in the vicinity of  $x=0$ :
That is, one cannot pass to the limit using  the Plemelj-Privalov theorem for all values
of the parameters involved 
in (\ref{eq:bd8}, \ref{eq:bd10}).
This is why we will develop another approach to give
a meaning to the limit value.

\section{The space $ \mathbb{X}^{\theta,\nu}$ ($\nu\neq 0$)}
\label{sec:5}

The solutions of the integral equations
evidently belong to a vectorial space 
of dimension two: see also (\ref{H10}).
In a first stage we will  
 design a particular basis in this space, 
in a second stage we will study  
the  properties of the two basis functions. 
A careful analysis of this
 singularity will allow
to show that one basis function (more precisely the $x$ component
of electric field)
is the sum
of 
a singular part
 $\frac1{\alpha(x)+i\nu}$ plus a term which
is bounded in $L^p$ ($1\leq p<\infty$)
uniformly with respect to $\nu$.
It will be the central result of this part.

For the simplicity of notations,
we  restrict the parameter to  $0<\nu\leq 1$ without loss
of generality. The extension to negative $\nu$ will be considered
in section (\ref{sec:x-}).
 We   define the vectorial
space of all solutions of the X-mode equations 
\begin{equation} \label{eq:bd26}
\mathbb{X}^{\theta,\nu}=
\left\{
x\mapsto \left(
U(x), V(x), W(x)\right),
\; \mbox{ for all solutions of the system }
(\ref{sys0:hatmu})
\right\}.
\end{equation}
One may also use the notation:
$
\mathbf{U}^{\theta,\nu}=
( U^{\theta,\nu}, V^{\theta,\nu},W^{\theta,\nu})\in  \mathbb{X}^{\theta,\nu}
$. 
This section is devoted to the analysis of this space.

\begin{rmk}
The property that 
$
\mbox{dim }\mathbb{X}^{\theta,\nu}=2
$
is also evident considering the  right hand side
of the integral equation (\ref{eq:bdeq:bd13}).
\end{rmk}

By elimination  $  U^{\theta,\nu} $
in (\ref{sys0:hatmu}), one gets a
system of two coupled ordinary differential  equations
\begin{equation} \label{H10}
\frac{d}{dx}
\left(
\begin{array}{c}
 V^{\theta,\nu} \\
 W^{\theta,\nu}
\end{array}
\right)=
A^{\theta,\nu}(x)
\left(
\begin{array}{c}
 V^{\theta,\nu} \\
 W^{\theta,\nu}
\end{array}
\right)
\end{equation}
 with 
\begin{equation} \label{H10.2}
{ A^{\theta,\nu}(x) }=
{
\left(
\begin{array}{cc}
\frac{\theta \delta(x) }{\alpha(x)+i\nu} &
1-\frac{\theta^2}{\alpha(x)+i\nu}
\\
\frac{\delta(x)^2}{\alpha(x)+i\nu}-
\alpha(x)-i\nu
 &
-\frac{\theta \delta(x) }{\alpha(x)+i\nu} 
\end{array}
\right)}
.
\end{equation}
In the case $\nu\neq 0$ the matrix is non singular for all 
$x$, which gives a meaning to the regularized problem.
One notices  the matrix  is singular for $\nu=0$.

\begin{lem}\label{lem:indep}
Take two  solutions 
$\left(V^{\theta,\nu}, 
W^{\theta,\nu}   \right)$ and
$\left(\widetilde{V}^{\theta,\nu}, 
\widetilde{W}^{\theta,\nu}   \right)$ 
of (\ref{H10}). Define the   Wronskian
\begin{equation} \label{H14}
{\cal W}(x)=  V^{\theta,\nu}(x)\widetilde{W}^{\theta,\nu}(x)-
{W}^{\theta,\nu}(x) \widetilde{V}^{\theta,\nu}(x).
\end{equation} 
Then the Wronskian is constant:
$
{\cal W}(x)= {\cal W}(0)$ for all $ x$.
\end{lem}
\begin{proof}
The system (\ref{H10}) main be rewritten as
$$
\frac{d}{dx}
\left(
\begin{array}{c}
V \\
 W
\end{array}
\right)=
\left(
\begin{array}{cc}
a & b
\\
c
 &
-a
\end{array}
\right)
\left(
\begin{array}{c}
V \\
W
\end{array}
\right).
$$
Therefore
$$
\frac{d}{dx}{\cal W}=
\frac{d}{dx}\left(V(x)\widetilde W(x) - W(x) \widetilde V(x)    \right)
$$
$$
=
\left(a V+ b W   \right) \widetilde W +
V \left(c \widetilde V - a \widetilde W\right)
-
\left(c  V - a  W\right) \widetilde V
-W
\left(a \widetilde V+ b \widetilde W   \right) =0
$$
since all terms cancel each other.
\end{proof}

\subsection{The first basis function}

Next we desire to particularize a convenient basis
in this space.
The first  basis function 
\begin{equation} \label{H17.2}
\mathbf{U}_1^{\theta,\nu}
=
\left( 
U_1^{\theta,\nu}, V_1^{\theta,\nu},W_1^{\theta,\nu}
\right)
\in \mathbb{X}^{\theta,\nu},
\quad 
U_1^{\theta,\nu}(0)=0
\end{equation}
 is the natural one which
is smooth at the origin. For that reason $G$ is chosen to be the origin in this subsection,
 so that the corresponding integral equation has a bounded
right-hand side and a bounded kernel.
It is naturally characterized by
\begin{equation} \label{H17}
V_1^{\theta,\nu}(0)=  i \theta ,
\qquad \mbox{ and }
 W_1^{\theta,\nu}(0)=  i \delta(0) \qquad (\neq 0).
\end{equation}

\begin{prop}\label{prop:bbf1}
The  basis function (\ref{H17.2})
 is uniformly bounded with respect to
$\nu$:
for any interval $\theta \in [\theta_-,\theta_+]$ and any $H\in 
]L, \infty[$, there exists a constant
independent of $\nu$
such that
\begin{equation} \label{H18}
\left\| U_1^{\theta,\nu}\right\|_{ L^\infty(-L, H) } 
+\left\|  V_1^{\theta,\nu}\right\|_{ L^\infty(-L, H) } 
+
\left\| W_1^{\theta,\nu}
\right\|_{ L^\infty(-L, H) } \leq C.
\end{equation}
\end{prop}
\begin{proof}
The right hand side  in the 
integral equation 
(\ref{eq:bdeq:bd13})
is 
$$
g^\nu(x)=\frac{h^\nu(x)}{\alpha(x)+i\nu}
\mbox{ with  }h^\nu(x)=
 i \theta 
\Dx  A_\nu(x)+
 i \delta(0) 
\Dx  B_\nu(x).
$$
With the choice (\ref{H17.2}) one has 
$
h^\nu(0)=
i \theta 
(-i\delta(0))+
i \delta(0) 
(i\theta)=0 $ for all $ \nu$. 
Therefore the right hand side of the integral equation, namely 
$$
g^\nu(x)=\frac{h^\nu(x)-h^\nu(0)}{\alpha(x)+i\nu},
$$
is bounded around $0$. As it is moreover bounded away from $0$, it is bounded in $ L^\infty(-L, H)$ uniformly with respect
to $\nu$.
The solution $U_1^{\theta,\nu}$ (\ref{eq:bd10}) is also bounded, since by the results of Subsection 4.2 the kernel 
$ \K(x,z,0)$ is also uniformly bounded.  
These bounds are uniform with respect to $\nu$. The integral representation
(\ref{eq:bd11})
of the $V_1^{\theta,\nu}$
yields that $V_1^{\theta,\nu}$ is also bounded.
It is similar concerning
the integral  representation
(\ref{eq:bd12})
of the $W_1^{\theta,\nu}$,
so  $W_1^{\theta,\nu}$ is also bounded.
\end{proof}

\subsection{Behavior at infinity}\label{subsec:infty}

Hypothesis \eqref{H3} allows to study a simplified model  with constant 
coefficients
 for $x\geq H$. In fact, it corresponds to a system as in (\ref{H10}) with constant coefficients, which matrix will be denoted $A_\infty^{\theta,\nu}$. 

\begin{prop}
The matrix $A_\infty^{\theta,\nu}$ has two distinct eigenvalues.
The first eigenvalue  $\lambda^{\theta,\nu}$ has a positive
real part.
The second eigenvalue is $-\lambda^{\theta,\nu}$.
\end{prop}
\begin{proof}
The  eigenvalues
are solution to the characteristic equation
$$
\lambda^2-\mbox{tr}(A^{\theta,\nu}_\infty)\lambda
+\mbox{det}(A^{\theta,\nu}_\infty)=0
$$ 
where $\mbox{tr}(A^{\theta,\nu}_\infty)=0$ and
$
\mbox{det}(A^{\theta,\nu}_\infty)=
\alpha_\infty+i\nu-\theta^2-\frac{ \delta_\infty^2 }{ \alpha_\infty+i\nu }$. 
The real part is 
$$
\mbox{real}\left( \mbox{det}(A^{\theta,\nu}_\infty)  \right)=
\alpha_\infty-\theta^2-\frac{ \delta_\infty \alpha_\infty}{ \alpha_\infty^2+\nu^2 }
=
\alpha_\infty\left(
1-\frac{ \delta_\infty^2 }{ \alpha_\infty^2+\nu^2 }
\right)
-\theta^2
$$
and is therefore negative 
due to the coercivity assumption
(\ref{H4}).
So the usual square root  
$
\lambda^{\theta,\nu}=\sqrt{-  \mbox{det}(A^{\theta,\nu}_\infty)}
$
has a positive real part. 
The other one has a negative real part.
\end{proof}

As a consequence 
 any $\mathbf{U}\in\mathbb{X}^{\theta,\nu}$
is at large scale a linear combination of the exponential increasing
function and a exponential decreasing function
\begin{equation} \label{eq:bd43}
\mathbf{U}(x)=c_+ R_+e^{\lambda^{\theta,\nu}x  }+
c_- R_-e^{-\lambda^{\theta,\nu}x  } \qquad H\leq x 
\end{equation}
where $R_+\in \mathbb{C}^3$ and
 $R_-\in \mathbb{C}^3$
are constant vectors and $(c_+,c_-)\in \mathbb{C}^2$ are arbitrary
complex numbers.
Regarding the structure of the matrix and using the second
equation  of the system (\ref{sys0:hatmu}), one gets that
$R_+=(r_+^1,r_+^2,r_+^3)$ with 
$$
r_+^1=\frac{i\theta
r_+^3
-i\delta(H) r_+^2
  }{\alpha(H)+i\nu}, \;
r_+^2=1-\frac{\theta^2}{\alpha(H)+i\nu},
\;
r_+^3=\sqrt{ - \mbox{det}(A^{\theta,\nu}_\infty)}-
\frac{\theta \delta(H) }{\alpha(H)+i\nu}.
$$ 
The other vector 
$R_-=(r_-^1,r_-^2,r_-^3)$ is characterized by
$$
r_-^1=\frac{i\theta
r_-^3
-i\delta(H) r_-^2
  }{\alpha(H)+i\nu}, \;
r_-^2=1-\frac{\theta^2}{\alpha(H)+i\nu},
\;
r_-^3=- \sqrt{-  \mbox{det}(A^{\theta,\nu}_\infty)}-
\frac{\theta \delta(H) }{\alpha(H)+i\nu}.
$$ 
One notices that $R_+$ and $R_-$ are well defined 
for all  $\nu \in \mathbb{R}$, in particular  even for $\nu=0$.

\begin{prop}
The first basis function 
(\ref{H17.2})
is exponentially  growing at large scale ($\nu\neq 0$).
\end{prop}
\begin{proof}
For the sake of simplicity, denote $\mathbf{U}_1^{\theta,\nu}
=
\left( 
U_1, V_1,W_1
\right)$, dropping the $\theta$s and $\nu$s.
Then from system (\ref{sys0:hatmu}) one gets
$$
\left\{
\begin{array}{cccr}
W_1 & + i\theta U_1 &- V_1' &=0, \\
 i\theta  W_1 & -(\alpha +i\nu) U_1
 & -i\delta  V_1  &=0,  \\
- W_1'& + i\delta U_1 &-(\alpha+i\nu) V_1 &=0.
\end{array}
\right.
$$
Multiplying the second equation 
by $\overline{U_1}$ and the third one  by  $\overline{V_1}$, the sum writes 
$$
i\theta W_1 \overline{U_1 }-
W_1' \overline{V_1 }-
\left( \alpha  |U_1|^2+\alpha|V_1|^2  
+ i \delta V_1 \overline{U_1 }-
i \delta U_1 \overline{V_1 }
 \right)- i \nu \left(   |U_1|^2+|V_1|^2  \right)=0.
$$
On the other hand
an integration in the interval $]M,N[$  yields
$$
\int_M^N
\left( i\theta W_1 \overline{U_1 }-
W_1' \overline{V_1 } \right) dx
=
\int_M^N
\left( i\theta W_1 \overline{U_1 }+
W_1 \overline{V_1 }' \right) dx
- W_1(N)  \overline{V_1 }(N)
+  W_1(M)  \overline{V_1 }(M)
$$
$$
= \int_M^N
|W_1|^2   dx
- W_1(N)  \overline{V_1 }(N)
+  W_1(M)  \overline{V_1 }(M),
$$
where we used  the first equation. 
We obtain the identity,
\begin{equation} \label{eq:bd55}
\int_M^N
\left( |W_1|^2 -  \alpha  |U_1|^2-\alpha|V_1|^2  
- i \delta V_1 \overline{U_1 }+
i \delta U_1 \overline{V_1 }
   \right) dx
-i \nu 
\int_M^N  \left(   |U_1|^2+|V_1|^2  \right)dx
\end{equation}
$$
=W_1(N)  \overline{V_1 }(N)
-  W_1(M)  \overline{V_1 }(M).
$$
Splitting between the real and imaginary parts, one gets
the important relation
\begin{equation} \label{eq:bd45}
\nu \int_M^N  \left(   |U_1|^2+|V_1|^2  \right)dx=
\mbox{Im}\left(   W_1(M)  \overline{V_1 }(M) \right) -
\mbox{Im}\left(  
W_1(N)  \overline{V_1 }(N)
\right)
\end{equation}
which is true in fact
for any element in $\mathbb{X}^{\theta,\nu}$ and for any  $M<N$.

Let us take $M=0$: so $V_1(0)=\frac{\theta}{\delta(0)}W_1(0)$
and $\mbox{Im}\left(  
W_1(0)  \overline{V_1 }(0)
\right)=0$.
Therefore
$
\nu \int_0^N  \left(   |U_1|^2+|V_1|^2  \right)dx=
 -
\mbox{Im}\left(  
W_1(N)  \overline{V_1 }(N)
\right)$. 
It  shows that $W_1(N)  \overline{V_1 }(N)\not \rightarrow
 0$ for $N\rightarrow \infty$.
In other words the first basis function does not decrease
exponentially at infinity.
 Considering (\ref{eq:bd43}) it means that this function
is exponentially increasing at infinity.
\end{proof}

\subsection{The second basis function}

The second
basis function
$$
\mathbf{U}_2^{\theta,\nu}=
(U_2^{\theta,\nu},V_2^{\theta,\nu},W_2^{\theta,\nu})
\in \mathbb{X}^{\theta,\nu}
$$
is built with two requirements.
\begin{itemize}
\item  It is  exponentially decreasing at infinity:
there exists  $c_-\in \mathbb{C}$ such that 
\begin{equation} \label{eq:bd46}
\mathbf{U}_2^{\theta,\nu}(x)=c_-
R_-e^{-\lambda^{\theta,\nu}x  }, \qquad H\leq x ,
\end{equation}
\item 
Its value at the origin is  normalized with the requirement
\begin{equation} \label{eq:bd47}
i\nu U_2^{\theta,\nu}(0)=1.
\end{equation}
\end{itemize}
To ensure that these conditions can be satisfied, consider the third function 
\begin{equation} \label{eq:bd46.2}
\mathbf{U}_3^{\theta,\nu}=(U_3^{\theta,\nu},V_3^{\theta,\nu},W_3^{\theta,\nu})(x)=
R_-e^{-\lambda^{\theta,\nu}x  } \qquad H\leq x ,
\end{equation}
where $R_-$ and $\lambda_-$ are defined in Section \ref{subsec:infty}, smoothly extended so that  $
\mathbf{U}_3^{\theta,\nu} \in \mathbb{X}^{\theta,\nu}
$. The identity 
\begin{equation} \label{eq:dieze}
\nu \int_M^N  \left(   |U_3^{\theta,\nu}|^2+|V_3^{\theta,\nu}|^2  \right)dx
=
\mbox{Im}\left(   W_3^{\theta,\nu}(M)  \overline{V_3^{\theta,\nu} }(M) \right) -
\mbox{Im}\left(  
W_3^{\theta,\nu}(N)  \overline{V_3^{\theta,\nu} }(N)
\right)
\end{equation}
with $N\rightarrow \infty$ and $M=0$ shows that
$$
\nu \int_0^\infty  \left(   |U_3^{\theta,\nu}|^2+|V_3^{\theta,\nu}|^2  \right)dx=
\mbox{Im}\left(   W_3^{\theta,\nu}(0)  \overline{V_3^{\theta,\nu} }(0) \right) .
$$
However, from \eqref{sys0:hatmu}, $  V_3^{\theta,\nu}(0) 
=\frac{\theta}{\delta(0)}W_3^{\theta,\nu}(0) -\frac{\nu}{\delta(0)}U_3^{\theta,\nu}(0)$, 
so
one gets
\begin{equation} \label{eq:delta0}
\nu \int_0^\infty  \left(   |U_3^{\theta,\nu}|^2+|V_3^{\theta,\nu}|^2  \right)dx=
-\frac{\nu}{\delta(0)}
\mbox{Im}\left(   W_3^{\theta,\nu}(0)
  \overline{U_3^{\theta,\nu} }(0) \right).
\end{equation}
Since $\delta(0)\neq 0$ which is a major hypothesis
in our work, this shows that $U_3^{\theta,\nu}(0)\neq 0$.
This is why it is always possible to renormalize  with a parameter
\begin{equation} \label{eq:bd47.6}
\mathbf{U}_2^{\theta,\nu}
=c_-
\mathbf{U}_3^{\theta,\nu}, \qquad
c_-=\frac1{i \nu  U_3^{\theta,\nu}(0)}
\end{equation}
so as to 
enforce (\ref{eq:bd47}).

\begin{prop}
With the normalizations (\ref{H17}) and
(\ref{eq:bd46}-\ref{eq:bd47}),
the Wronskian relation  takes the form
\begin{equation} \label{eq:bd48}
V_1^{\theta,\nu}(x)W_2^{\theta,\nu}(x)-
W_1^{\theta,\nu}(x)V_2^{\theta,\nu}(x)=1 \qquad \forall x.
\end{equation}
\end{prop}

\begin{proof}
It is sufficient to compute it at the origin
$$
V_1^{\theta,\nu}(0)W_2^{\theta,\nu}(0)-
W_1^{\theta,\nu}(0)V_2^{\theta,\nu}(0)=
i\theta W_2^{\theta,\nu}(0) - i \delta(0)V_2^{\theta,\nu}(0)
$$
$$
=
(\alpha(0)+i\nu)U_2^{\theta,\nu}(0)=i\nu U_2^{\theta,\nu}(0)=1
$$
using (\ref{sys0:hatmu}) and  thanks to (\ref{eq:bd47}).
\end{proof}

\begin{rmk}
The value of the Wronskian  (\ref{eq:bd48}) is independent 
 of $\nu$. It will be  
of  major interest in the limit   regime $\nu\rightarrow 0$.
\end{rmk}

The non zero Wronskian shows (\ref{eq:bd48}) shows
that the two basis function are linearly independent.
So they span the whole space
$$
\mathbb{X}^{\theta,\nu}
=\mbox{Span}\left\{
\mathbf{U}_1^{\theta,\nu},\mathbf{U}_2^{\theta,\nu}
\right\}, \qquad \nu >0.
$$

\subsection{Passing to the limit $\nu\rightarrow 0$}

We now study 
the limit $\nu\rightarrow 0$.
An important result is that the first basis function
admits a limit which is defined as a continuous
function in ${\cal C}^0[-L,\infty[$ and is independent of the sign
of $\nu$.
On the other hand the second basis function admits
a limit which is singular at $x=0$. Moreover the 
limit is different for 
 $\nu\rightarrow 0^+$ and for 
$\nu\rightarrow 0^-$. 
The linear independence of these limits will be establish
with a transversality condition.

\subsubsection{The first basis function}

There is no difficulty for this case which is easily
treated passing to the limit in the integral equation
(\ref{eq:bd10}), choosing $G=0$.
The limit basis function is referred to as
$$
\mathbf{U}_1^{\theta}=
(U_1^{\theta},V_1^{\theta},W_1^{\theta})
$$
 $\mathbf{U}_1^{\theta}$ is and will be called the regular solution by analogy with the terminology in scattering on the half-line. 
It is defined as the solution of a limit version 
of \eqref{eq:bdeq:bd13}, the $V$ and $W$ component being defined by limit versions of \eqref{eq:bd11} and \eqref{eq:bd12}:
 \sysnn{l}{ 
\displaystyle U_1^\theta (x) - \int_{0}^x \bar K^\theta(x,z) U_1^\theta (z) dz = \bar F^{\theta}(x),
\\\displaystyle  V_1^\theta (x) = i\theta A(x)+i\delta(0)B(x) +  \int_{0}^x \Dz k(x,z)U_1^\theta (z)dz ,
\\\displaystyle  W_1^\theta (x) =i\theta A'(x)+ i\delta(0)B'(x) +  \int_{0}^x \partial_x \Dz k(x,z)U_1^\theta (z) dz ,
}
where
\enn{
\bar K^\theta(x,z) =
\left\{\tab{ll}{
\displaystyle
 \frac{\Dx \Dz k(x,z)}{\alpha(x)}  & \forall x\neq 0 \textrm{ and  }0 \leq z \leq x \textrm{ or } x\leq z \leq 0,
\\ 0&\textrm{ in all other cases},
}\right. 
 }
is the limit kernel described in Proposition \ref{prop:limker} and
\enn{
\bar F^{\theta}(x) =
\left\{\tab{ll}{
\displaystyle\frac{i\theta \Dx  A(x)+i\delta(0) \Dx  B(x) }{\alpha(x)} & \forall x\neq 0,
\\ 
\displaystyle
\frac{\left(i\theta \Dx  A+i\delta(0) \Dx  B\right)'(0)}{\alpha '(0)} & \textrm{otherwise}.
}\right.  }
The right hand side $\bar F^{\theta}$ together with the kernel $\bar K^\theta$ considered in the integration domain are continuous, because $ \Dx  A(0)=
-i\delta(0)$, $ \Dx  B(0)=i\theta$ and see Proposition \ref{prop:limker}.

A preliminary pointwise convergence will be used to obtain an $L^p$ convergence result.
\begin{lem}\label{lem:umu_u}
There is  pointwise convergence of the first component
$$
\left\|\left(U_1^{\theta,\nu}(x)-\frac{F^{\theta,\nu}(x)}{\alpha(x)+i\nu}\right)-\left(U_1^{\theta}-\bar F^\theta\right) (x)\right\|_{L^\infty(]-L,H[)}\rightarrow 0
$$
which yields $\left\|U_1^{\theta,\nu}-U_1^{\theta} \right\|_{L^\infty_{\rm loc} (]-L,0[\cup]0,H[)}\rightarrow 0$. 

As a result the other components satisfy
$$
\left\|V_1^{\theta,\nu}-V_1^{\theta} \right\|_{L^\infty(]-L,H[)}\rightarrow 0,
\mbox{ and }
 \left\|W_1^{\theta,\nu}-W_1^{\theta} \right\|_{L^\infty(]-L,H[)}\rightarrow 0.
$$
\end{lem}

\begin{proof}
{\bf Convergence away from zero}

From the integral equations satisfied by $U_1^{\theta,\nu}$ and $U_1^{\theta}$ one has for all $x\in(-L,\infty)$ and all $\nu \neq 0$ the following integral equation on $ U_1^{\theta,\nu} - U_1^{\theta}$:
\e{ \label{eq:umu-u}
\tab{l}{
\displaystyle \left(U_1^{\theta,\nu}- U_1^{\theta}\right) (x)-\int_0^x \frac{\Dx \Dz k^{\nu}(x,z)}{\alpha(x)+i\nu}\left( U_1^{\theta,\nu} - U_1^{\theta} \right) (z) dz
\\\displaystyle \phantom{U_1^{\theta,\nu} (x)-}
 =\underbrace{ \frac{F^{\theta,\nu}(x)}{\alpha(x)+i\nu}-\bar F^\theta(x)}_{T_1} +\int_0^x \left(\underbrace{ \frac{\Dx \Dz k^{\nu}(x,z)}{\alpha(x)+i\nu} - \bar K(x,z)}_{T_2} \right) U_1^{\theta}(z) dz  .
}}
Since the kernel of equation \eqref{eq:umu-u} is bounded, the resolvent kernel $\mathcal K^{\theta,\nu}$ is bounded, see Remark \ref{rk:kerb}.

Denote $\mathcal F_\nu$ the right hand side of equation \eqref{eq:umu-u}.
Since $F^{\theta,\nu}(0)=0$ and $\Dx \Dz k^{\nu}(0,0)=0$, then $\mathcal F_\nu$ is \underline{bounded} on $]-L,H[$.

The $T_1$ term converges pointwise to 0 at any $x\neq 0$ thanks to the definition of $\bar F^\theta$.
Since $T_2$ pointwise converges to $0$ and because it is bounded as indicated in Remark \ref{rk:kerb}, the dominated convergence theorem shows that the integral term in $\mathcal F_\nu$ pointwise converges to $0$ as long as $x\neq 0$ - note that it is obviously true for $x=0$. Thus $\mathcal F_\nu$ \underline{pointwise converges} to $0$ as long as $x\neq 0$.

As a result, the dominated convergence theorem shows that 
\enn{ \left| U_1^{\theta,\nu} (x)- U_1^{\theta} (x)\right| 
\leq 
\left| \mathcal F_\nu(x)\right| + \left\| \mathcal K^{\theta,\nu} (x,z)\right\|_{L^\infty(\mathcal D_0 \cap \{ x\in]-L,H[ \})} \int_0^x\left| \mathcal F_\nu(z)\right|dz}
pointwise converges to zero as long as $x\neq 0$ as well.

Note that at $x=0$, \eqref{eq:umu-u}  reads
$
U_1^{\theta,\nu} (0)- U_1^{\theta} (0)=\frac{F^{\theta,\nu}(0)}{i\nu}-\bar F^\theta(0) = -\bar F^\theta(0) $.
Then, if  $\bar F^\theta(0) = 
\delta'(0) \theta + \delta(0)^2\neq 0$
the pointwise convergence of $ U_1^{\theta,\nu} - U_1^{\theta}$ at $x=0$ 
does not hold. Indeed, the term $\bar F^\theta(0) $ does not depend on $\nu$. 
However, if $\delta'(0) \theta + \delta(0)^2=  0$ we have pointwise
convergence at $x=0$ since in this case 
$U_1^{\theta,\nu} (0)- U_1^{\theta} (0)=0$ for all $\nu$.

{\bf Convergence on $]-L,H[$}

Despite the last remark, a convergence in $L^\infty(]-L,H[)$ can be obtained subtracting the appropriate quantities to the first component and its limit. 
By (\ref{eq:umu-u})
$$
\displaystyle 
\left| 
\left( 
\left(U_1^{\theta,\nu}- \frac{F^{\theta,\nu}(x)}{\alpha(x)+i\nu}\right)
- 
\left( U_1^{\theta}\right) (x) -\bar F^\theta(x)
\right)
\right|
$$
$$
\leq 
\int_0^x \left| 
\frac{\Dx \Dz k^{\nu}(x,z)}{\alpha(x)+i\nu}\right|
\left|  U_1^{\theta,\nu} - U_1^{\theta} \right| (z) dz
 +\int_0^x \left|
\frac{\Dx \Dz k^{\nu}(x,z)}{\alpha(x)+i\nu} - \bar K(x,z)
 \right| \left|  U_1^{\theta}(z)\right|  dz 
$$
Then, by the dominated convergence theorem, 
 the function
$
\left(U_1^{\theta,\nu}-\frac{F^{\theta,\nu}}{\alpha+i\nu}\right)-\left(U_1^{\theta}-\bar F^\theta\right)
$
converges to zero in $L^\infty(]-L,H[)$.

The convergence of $V_1^{\theta,\nu}$ and $W_1^{\theta,\nu}$ then 
stems from the dominated convergence theorem again. Indeed, since 
 \sysnn{ll}{ 
\displaystyle  V_1^{\theta,\nu} (x)-V_1^\theta (x) &
= i\theta (A_\nu-A)(x)+i\delta(0)(B_\nu-B)(x) 
\\ &\displaystyle \phantom = +  \int_{0}^x \left(\Dz k^\nu(x,z)U_1^{\theta,\nu} (z)- \Dz k(x,z)U_1^\theta (z) \right)dz ,
\\\displaystyle  W_1^{\theta,\nu} (x)-W_1^\theta (x) &=i\theta (A_\nu-A)'(x)+
 i\delta(0)(B_\nu-B)'(x) 
\\ &\displaystyle \phantom = +  \int_{0}^x \left( \partial_x \Dz k^\nu(x,z)U_1^{\theta,\nu}-\partial_x \Dz k(x,z)U_1^\theta (z)\right) dz ,
}
the $L^\infty$ convergence of both terms $\Dz k^\nu(x,z)U_1^{\theta,\nu} (z)- \Dz k(x,z)U_1^\theta (z) $ and $ \partial_x \Dz k^\nu(x,z)U_1^{\theta,\nu}-\partial_x \Dz k(x,z)U_1^\theta (z)$ on $]-L,0[$ and $]0,H[$ ensures that 
the hypothesis of the dominated convergence theorem are satisfied. The convergence then holds on $]-L,H[$ since at $x=0$ it is guaranteed by the convergence of $A_\nu$ and $B_\nu$.
\end{proof}

\begin{prop}
The first basis functions satisfies 
\enn{\left\| \mathbf U_1^{\theta,\nu} -  \mathbf U_1^{\theta} \right\|_{L^p(-L,H)} \rightarrow 0,
\qquad 1\leq p < \infty. 
}
\end{prop}
\begin{proof}
The $L^1$ convergence is a consequence of the pointwise convergence obtained in Lemma \ref{lem:umu_u} thanks to the dominated convergence theorem.
Moreover Proposition \ref{prop:bbf1} yields an $L^\infty $ bound for $ \mathbf U_1^{\theta,\nu} -  \mathbf U_1^{\theta} $.
The result is thus straightforward.
\end{proof}

The next result establishes that 
$\mathbf{U}_1^\theta$ is still
exponentially increasing at infinity with a technical condition.

\begin{prop}
Assume hypothesis \eqref{H5}. 
Then 
$\mathbf{U}_1^{\theta=0}$ 
 increases exponentially at infinity. 
\end{prop}
\begin{rmk}
The constant 4 in the condition \eqref{H5} is probably non optimal.
\end{rmk}

\begin{proof}
We drop the super-index $\cdot ^{\theta=0}$ to simplify:
that is $(U_1,V_1,W_1)$ stands for $(U_1^{0},V_1^{0},W_1^{0})$.
Let us consider the identity (\ref{eq:bd55}) which holds true
at the limit $\nu=0$
$$
\int_0^N
\left( |W_1|^2 -  \alpha  |U_1|^2-\alpha|V_1|^2  
- i \delta V_1 \overline{U_1 }+
i \delta U_1 \overline{V_1 }
   \right) dx
$$
$$
=W_1(N)  \overline{V_1 }(N)
- W_1(0)  \overline{V_1 }(0), \quad 0<N<\infty.
$$
Since we consider the case $\theta=0$, 
$V_1(0)=0$.
Notice also that $W_1=V_1'$, so the relation is rewritten as
$$
\int_0^N
\left( |V_1'|^2 -  \alpha  |U_1|^2-\alpha|V_1|^2  
- i \delta V_1 \overline{U_1 }+
i \delta U_1 \overline{V_1 }
   \right) dx=W_1(N)  \overline{V_1 }(N).
$$
Let us proceed by contradiction:  we assume that the  function
 is exponentially decreasing at infinity.
It yields
$$
\int_0^\infty
\left( |V_1'|^2 -  \alpha  |U_1|^2-\alpha |V_1|^2  
- i \delta V_1 \overline{U_1 }+
i \delta U_1 \overline{V_1 }
   \right) dx=0.
$$ 
Notice that 
$-  \alpha  |U_1|^2-\alpha |V_1|^2  
- i \delta V_1 \overline{U_1 }+
i \delta U_1 \overline{V_1 }\geq 0$ for $x\geq H$ due to the coercivity property
(\ref{H4}).
Therefore it implies that
$$
\int_0^H
\left( |V_1'|^2 -  \alpha  |U_1|^2-\alpha |V_1|^2  
- i \delta V_1 \overline{U_1 }+
i \delta U_1 \overline{V_1 }
   \right) dx\leq 0.
$$
Next observe that $U_1=-i\frac{\delta}\alpha V_1$, so that
$
\int_0^H
\left( |V_1'|^2 +  \frac{\delta^2} \alpha  |V_1|^2-\alpha |V_1|^2  
   \right) dx\leq 0$. 
Since $V_1(0)=0$ and $\alpha(x)\approx \alpha'(0)x$ with
$\alpha'(0)<0$ (see hypothesis \ref{H1}), 
 it is convenient to notice the proximity with the 
famous Hardy inequality that we recall, 
\begin{equation} \label{eq:bd60}
\int_0^H \frac{u(x)^2}{x^2} < 4 \int_0^H u'(x)^2, \qquad
 u\in H^1(0,H), \quad  u(0)=0, \quad u\neq 0.
\end{equation} 
Since, thanks to hypothesis \eqref{H2},
$$
\int_0^H \frac{\delta^2}{|\alpha|}  |V_1|^2
=  \int_0^H \delta^2  x \frac{x}{|\alpha|}  \frac{|V_1|^2}{x^2}
\leq \frac{\left\|\delta\right\|_\infty^2 H  }{r}
\int_0^H  \frac{|V_1|^2}{x^2},
$$
it yields the inequality 
$$ 0\leq 
\left(1-4  \frac{\left\|\delta\right\|_\infty^2 H  }{r}  \right) \int_0^H
|V_1'|^2 dx \leq 
\int_0^H
\left( |V_1'|^2 +  \frac{\delta^2}\alpha  |V_1|^2-\alpha |V_1|^2  
   \right) dx\leq  0,
$$
 where we used
 \eqref{H5}. Therefore $V_1$ vanishes on the interval
$[0,H]$. So $U_1$  vanishes and  $W_1$ also vanishes on the interval
which is not compatible with $W_1(0)=i\delta(0)\neq 0$.

\end{proof}

\begin{prop} \label{prop:bd2}
There exists a maximal value 
$\theta_{\rm{thresh}}>0$
such  that: 
If
hypothesis \eqref{H5} 
 is satisfied 
and $|\theta| < \theta_{\rm{thresh}}$, then 
$\mathbf{U}_1^{\theta}$ 
 increases exponentially at infinity. 
\end{prop}

Let us denote by  $(U_3^\theta,V_3^\theta,W_3^\theta)$ the solution to \eqref{sys0:hatmu}  for $ x>0$
that satisfies the exponentially decreasing condition (\ref{eq:bd46.2}) with $\nu=0$.
\begin{proof}
Let us consider the function
\begin{equation} \label{eq:bd47.8}
\sigma(\theta)=V_1^\theta(H)W_3^\theta(H)
-
W_1^\theta(H)V_3^\theta(H)
\end{equation}

By definition 
$$
\left(
V_3(H),W_3(H)
\right)
=\left(1-\frac{\theta^2}{\alpha_\infty}  ,
-\frac{\theta \delta_\infty}{\alpha_\infty  }-
\sqrt{-\alpha_\infty+\theta^2+ 
\frac{\delta_\infty^2}{\alpha_\infty}    }\right)e^{-\sqrt{-\det A^{\theta,\nu=0}_\infty}H   }.
$$
This vector is real and always non  zero.
Therefore 
the function $\theta\mapsto f(\theta)$ is well defined. 
This function naturally satisfies two properties
\begin{itemize}
\item $\sigma(0)\neq 0$ since $(V_1^0,W_1^0)$ is exponentially increasing
by virtue of the previous
property. Indeed $\sigma(0)=0$ if and only if 
the functions  $x\mapsto (V_1^0(x),W_1^0(x))$ and 
$x\mapsto (V_3^0(x),W_3^0(x))$ are linearly dependent, which is not true.
\item the function $\sigma$ is continuous since the first basis function
is continuous with respect to $\theta$. 
\end{itemize} 
Therefore there exists an interval around 0 in which 
$\sigma(\theta)$ is non zero,
which in turn yields the fact that 
$\mathbf{U}_1^{\theta}$  is linearly independent
of $\mathbf{U}_3^{\theta}$. Therefore 
$\mathbf{U}_1^{\theta}$ 
is
exponentially increasing.
\end{proof}

\subsubsection{The transversality condition}

 Passing to the limit in the second  basis function
near the origin
is  involved. Indeed we expect that
the limit $U_2^\theta$ is such that
$ U_2^\theta\approx \frac{C}x $ for some local constant $C$.
 Therefore the limit is singular
and special care has to be provided to avoid
any artifacts in the analysis.

Let us define the special Wronskian
between the first and third basis functions
 $$ 
\sigma(\theta,\nu)=
V_1^{\theta,\nu}(H)W_3^{\theta,\nu}(H)- W_1^{\theta,\nu}(H)V_3^{\theta,\nu}(H). 
$$
It is the natural continuous extension with respect
to $\nu$ of the function $\theta\mapsto 
\sigma(\theta)$. 
We rewrite (\ref{eq:bd47.6}) as
$$
\mathbf{U}_2^{\theta,\nu}=
\xi^{\theta,\nu }
\mathbf{U}_3^{\theta,\nu}.
$$
Plugging this relation in the Wronskian
(\ref{eq:bd48}) one gets
that
$
1=\xi^{\theta,\nu} \sigma (\theta,\nu)$. 
This function is continuous with respect to $\nu$. Moreover the function
defined in (\ref{eq:bd47.8}) satisfies 
$\sigma(\theta)=\sigma(\theta,0)$. 
The transversality condition is 
defined as 
the condition
\begin{equation} \label{eq:bd63}
\sigma(\theta)\neq 0.
\end{equation}
 If the transversality
condition  is not satisfied, that is $\sigma(\theta)=0$,
 then by continuity 
$|\xi^{\theta,\nu}|\rightarrow \infty$ for $\nu\rightarrow 0$.
If $\sigma(\theta)=0$, then  the first  basis function and the third function
are   linearly dependent at the limit $\nu=0$.
It is of course possible
to develop the theory in this direction,
but it seems to us less interesting.
Therefore we will always assume the transversality condition\footnote{The 
"transversality condition" is a sufficient condition of linear
 independence.}
from now on.

\begin{prop}
Assume the transversality condition (\ref{eq:bd63}).
Then for all $\epsilon >0$ one has the limit
$$
\left\| \mathbf{U}_2^{\theta,\nu}- 
 \frac1{\sigma(\theta)}\mathbf{U}_3^\theta\right\|_{ 
\left(L^\infty[\epsilon,\infty[\right)^3 } 
 \rightarrow 0.
$$
\end{prop}

\begin{proof}
Evident.
\end{proof}

In order to show that the second basis function
admits a continuous limit for $x<0$, the  strategy
is to solve 
 the integral equation
(\ref{eq:bdeq:bd13}) 
from $G=H$ backward, and to show that fine
estimates on the solution give knowledge of the limit
even for $x<0$.

\begin{figure}[ht]
\begin{center}
\begin{tabular}{cc}
\includegraphics[width=10cm,angle=0]{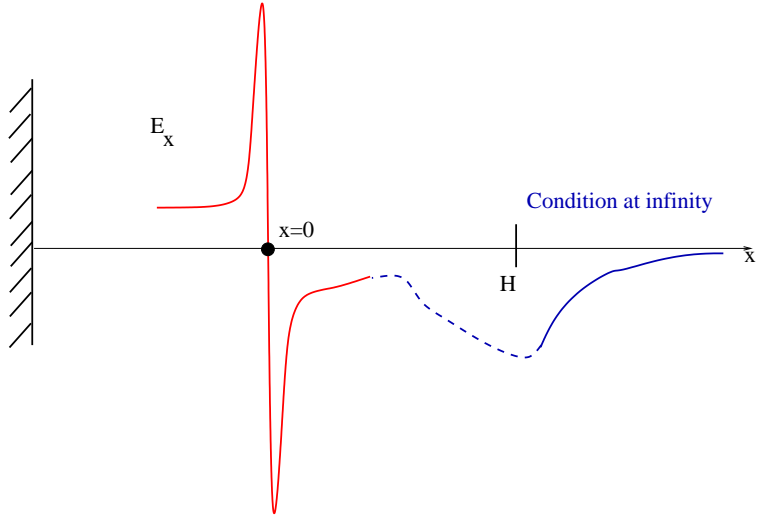}
\end{tabular}
\end{center}
\caption{Schematic representation of the real part of 
the limit electric field of the second basis function
$U_2^{\theta,\nu}$, $\nu>0$. Here the transversality  condition
$\sigma(\theta)\neq 0$
is satisfied, which turns into a singular behavior
at the limit $\nu\rightarrow 0$.}
\label{fig:elecinfini}
\end{figure}


\subsubsection{Continuity estimates}

The integral equation (\ref{eq:bdeq:bd13})  is singular at
the limit.
The whole problem comes form the singularity at $x=0$.
By comparison with the standard literature
\cite{Tricomi,musk,vekua,dautray:lions,picard,bart:warnock,shulaia}
we found no convenient mathematical tool to analyze
its properties.
That is why we develop in the following new continuity estimates
with respect to the parameters of the problem.
On this basis we will manage to pass to to the limit $\nu\rightarrow 0$.

Let us consider a general solution
$\mathbf{U}=(U,V,W)\in \mathbb{X}^{\theta,\nu}$
of the integral equation (\ref{eq:bdeq:bd13}) 
with  prescribed    data in  $H$ under the form
$$
V(H)=a_H
\mbox{ and }
W(H)=b_H.
$$
Let us introduce the compact notation
$$
\| H\| =\left| a_H \right|
+ \left| b_H \right|.
$$
Our goal is to obtain some sharp continuity estimates on the solution
$\mathbf{U}$ with respect to $\|H\|$.
The main point is to bound the constants uniformly with respect
to $0< \nu \leq 1 
$ which is hereafter taken positive for the simplicity of notation.
The reference point can be different from $H$ as well, but
non equal to zero. 
Once these continuity estimates are proved, they will provide enough 
 information
to define the limit $\nu\rightarrow 0$ of  the second basis function.

\begin{prop}\label{prop:boundsU} 
There exists a constant $C_\theta$ 
with continuous dependence with respect to 
$\theta$   such that
\begin{equation} \label{eq:bd71}
\left| U(x) \right|\leq \frac{C_\theta}{\sqrt{r^2 x^2+\nu^2}}
\| H\|,
\qquad 0<x\leq H.
\end{equation}
\end{prop}

\begin{proof}
Let us consider
$$
\gamma_\theta=\left(
\sup_{0\leq \nu \leq 1}
\|A_\nu\|_{ W^{1,\infty}(0,H) }+ \sup_{0\leq \nu \leq 1}
\|B_\nu\|_{ W^{1,\infty}(0,H)  }\right) (\|\delta\|_\infty+|\theta|).
$$
The integral equation
 (\ref{eq:bdeq:bd13})  with $G=H$ 
implies that
$$
\left| U(x) \right|
\leq \frac{ \gamma_\theta \| H\|}{
\sqrt{r^2 x^2+\nu^2} }
 + \int_x^H \frac{ |\Dx \Dz k (x,z)| }{ \sqrt{r^2 x^2+\nu^2}  } |U(z)|dz, 
$$
where we used  \eqref{H2}.
Since $\Dx \Dz k (x,x)=0$ for all $x$, there exists a constant 
$\beta_\theta$ such that  
$$
\left\| \Dx \Dz k (x,z) \right\|_{L^\infty]0,H[}
\leq \beta_\theta { |x-z| }
\leq \beta_\theta z \quad \mbox{ for  }0\leq x\leq z.
$$
So
$
 \sqrt{r^2 x^2+\nu^2} |U(x) |
\leq \gamma_\theta \|H\|+ \beta_\theta
\int_x^H z |U(z)|dz
$
and 
$$
r x  |U(x) |
\leq \gamma_\theta \|H\|+ \beta_\theta
\int_x^H z |U(z)|dz, \qquad 0\leq x\leq H.
$$
The Gronwall lemma is useful to study this inequality.
Indeed let us set  $g(x)=\int_{x}^H  |z U(z)|dz$,
 so that the previous inequality is rewritten as
$
-r  g'(x)\leq 
\gamma_\theta \| H\| + \beta_\theta g(x)$. 
Therefore $0\leq 
\gamma_\theta \| H\| 
+ 
r  g'(x)+ 
\beta_\theta g(x)$, that is:
$
0\leq \gamma_\theta \| H\| 
e^{\frac{\beta_\theta}rx   }
+
r
\left(e^{\frac{\beta_\theta}rx   }
g(x)
\right)'$. 
Next we  integrate on the interval $[x,H]$ and use the fact that
$g(H)=0$ by definition. It yields
$
0\leq \gamma_\theta \|H\|
\frac{ e^{\frac{\beta_\theta}rH} - 
e^{\frac{\beta_\theta}r x   } 
   }{\frac{\beta_\theta}r  }
-r e^{\frac{\beta_\theta}r x   } 
g(x)$, 
that is
\begin{equation}
g(x)\leq \frac{
e^{\frac{\beta_\theta}{r} (H-x)   }-1  
}{\beta_\theta}\; 
\gamma_\theta  \|H\|.
\end{equation}
Finally one checks that
$
 \sqrt{r^2 x^2+\nu^2} |U(x) |
\leq 
\gamma_\theta\|H\|+\beta_\theta g(x)
\leq 
e^{\frac{\beta_\theta}r (H-x)   } \gamma_\theta \|H\|
$
which proves (\ref{eq:bd71}).
\end{proof}

Next define
$\| 0\| =\left| V(0) \right|
+ \left| W(0)\right|$.

\begin{prop}\label{prop:maj0}
There exists a constant $C_\theta$ 
with continuous dependence with respect to
$\theta$  such that
\begin{equation}\label{eq:bd74}
\| 0 \|\leq  C_\theta (1+ |\ln \nu | ) \| H\|.
\end{equation}
\end{prop}
\begin{proof}
We adopt the same notations as above.
The integral expression of $V$ 
(\ref{eq:bd11}) with $G=H$
 yields the inequality
$$
\left|
V(0) \right|
\leq 
\gamma_\theta \left\| H \right\| 
 +  \left| \int_0^{H} \Dz k(0,z)
U(z)
\right| 
dz
$$
We notice that 
$
\Dz k(0,z)=\left(  \Dz k(0,z)-\Dz k(0,0)\right)
+ \Dz k (0,0)
$.
Since 
$
\Dz k(0,0)=i\theta \partial_z k(0,0)-i\delta k(0,0)=i\theta
$
one gets
$
\left|
\Dz k(0,z) -i \theta 
\right| \leq \eta_\theta  |z| 
$
for some constant $\eta_\theta>0$.
Which gives
$$
\left|
V(0) \right|
\leq 
\underbrace{\gamma_\theta \left\| H \right\| 
 +  \eta_\theta \int_0^{H} z
|U(z)| 
dz}_{Q}
+ |\theta| \underbrace{
\left| \int_0^H U(z) dz \right|}_{R}.
$$
By \eqref{eq:bd71}  $ Q \leq C_{\theta} \| H\|$, and moreover,

$$
R:=
\left| \int_0^H U(z) dz \right|
\leq  C_\theta \|H\|   \left| \int_0^H  \frac{1}{r|x|+\nu} dz \right| \leq  C_\theta \|H\| |\ln \nu|.
$$
 This completes  the proof  for $|V_2(0) |$.
 The  term  $|W(0) |$ is bounded with the same method
starting from the integral (\ref{eq:bd12}) and using the identity
$\partial_x \Dz k_\nu(x,x)= i \delta(z) A_\nu(z) $.
\end{proof}

An interesting question is the following.
Let us consider the integral equation
 (\ref{eq:bdeq:bd13})  with $G=0$.
That is the starting point of the integral is the singularity. 
One may wonder if a direct use of the Gronwall lemma 
may yield valuable estimates, or not.
It appears that a pollution with $\log \nu$ terms
render the result of little interest.

Consider firstly for simplicity $0\leq x$. Then  
 (\ref{eq:bdeq:bd13})
  with $G=0$ turns into
\begin{equation} \label{eq:pol2}
\left|U(x)\right| \leq 
C_\theta \frac{ \|0\| }{\sqrt{r^2 x^2 +\nu^2}}
+ C 
\int_0^x \left| U(z)\right|dz
\end{equation}
where we used (\ref{eq:ddfond}) to bound the kernel. The constant $C_\theta>0$ is 
 chosen large enough. Set $h(x)=\int_0^x \left| U(z)\right|dz$ so that
$
h'(x)\leq C_\theta
\frac{ \|0\| }{\sqrt{r^2 x^2 +\nu^2}}
+C_\theta h(x)$. 
Since $h(0)=0$ the Gronwall lemma yields the inequality
$
h(x)\leq C_\theta'\int_0^x \frac{ \|0\| }{ |z| +|\nu|}
dz
$
that is after integration ($0\leq x \leq H$)
$
\left| h(x)\right| \leq C_\theta'' \|0\| \left( 1+|\ln \nu|\right),
$
for some constant $C_\theta''>0$  with continuous dependence
with respect to $\theta$. Considering the bound (\ref{eq:bd74})
and the symmetry between $0<x$ and $x<0$ in the integral
 (\ref{eq:bdeq:bd13}) (with $G=0$) one obtains
the estimate
\begin{equation} \label{eq:pol}
\left| \int_0 ^x U(z)dz\right|
\leq C_\theta''' \|H\| \left( 1+|\ln\nu|\right)^2
, \quad -L\leq x \leq H.
\end{equation}
 Going back to (\ref{eq:pol2}) which is easily generalized to $x<0$, one gets 
\begin{equation} \label{eq:pol3}
\left|U(x)\right|\leq 
C_\theta\left(
\frac1{\sqrt{r^2x^2+\nu^2}}
+1+|\ln \nu| 
\right) \left( 1+|\ln \nu|   \right) \|H\|, \quad
-L\leq x \leq H.
\end{equation}
By  comparison of (\ref{eq:bd71}) and (\ref{eq:pol3}),
it is clear that this technique generates
 spurious terms of order $\log \nu$ 
for positive $x$. It spoils the possibility 
of having sharp estimates also for negative $x$.
With this respect, the rest of this section is devoted to 
the derivation of various sharp inequalities
which are free of such spurious terms.

Let us define
\begin{equation}\label{eq:bd112}
Q(\mathbf{U})= V_1^{\theta,\nu}(H)W(H)-W_1^{\theta,\nu}(H)V(H). 
\end{equation}
This quantity is  the Wronskian of the current solution
$\mathbf{U}$ against the first basis function.
It is therefore independent of the position $H$ which is used to evaluate
$Q(\mathbf{U})$.

\begin{prop}
There exists a constant $C_\theta$ 
with continuous dependence with respect to
$\theta$ and a continuous function $\nu \mapsto
\varepsilon(\nu)$ with $\varepsilon(0)=0$ 
 such that
\begin{equation} \label{eq:bd87}
\left| \;
|\nu| \left\|
U \right\|^2_{L^2(-L,H)}
- \left| \frac{ \pi Q(\mathbf{U})^2  }{\alpha'(0)}\right|
\; \right|
\leq C_\theta \varepsilon(\nu) \| H\|^2. 
\end{equation}
\end{prop}

\begin{proof} 
We consider positive $\nu$ to simplify the notations.
The proof is easily adapted for negative $\nu$.

Consider the integral equation (\ref{eq:bdeq:bd13})
with $G=0$. One gets
$$
U(x)=\frac{a_0\Dx A(x)+b_0\Dx B(x)}{\alpha(x)+i\nu}
+\int_0^x \frac{\Dx \Dz k(x,z)}{\alpha(x)+i\nu }
   U(z)dz.
$$
Here $(a_0,b_0)$ are a priori different from
$(a_H,b_H)$.
Due to (\ref{eq:bd112}), the normalization of $\mathbf U_1$ and thanks 
to Lemma \ref{lem:indep} one has that 
$
a_0\Dx A(0)+b_0\Dx B(0)=  Q(\mathbf{U})$. 
So the integral equation  can be written as
\begin{equation} \label{eq:bd83}
U(x)=
\underbrace{\frac{ Q(\mathbf{U})  }{\alpha(x)+i\nu}}_{S_1}
\end{equation}
$$
+
\underbrace{
 a_0 \frac{\Dx A(x)-\Dx A (0)}{\alpha(x)+i\nu}
+
b_0 \frac{\Dx B(x)-\Dx B (0)}{\alpha(x)+i\nu}
}_{S_2}
+
\underbrace{
\int_0^x \frac{\Dx \Dz k(x,z)}{\alpha(x)+i\nu }  U(z)dz
}_{S_3}.
$$
\begin{itemize}
\item The $L^2$ norm of the first term $S_1$  depends upon the 
 value of 
$$
D_\nu=\int_{-L}^H
\frac{\nu}{\alpha(x)^2+\nu^2}dx.
$$
Make the change of variable
$x=\nu w$ so that
$
D_\nu=\int_{-\frac{L}\nu}^{\frac{H}\nu}
\frac{1}{ 
b_\nu(w) ^2+1}dw$ and 
$b_\nu(w)=\frac{\alpha(\nu w)}\nu $.
Using the hypothesis 
(\ref{H2}) one has that $| b_\nu( w ) | \geq r w$, $r>0$. Since
$\int_\mathbb{R} \frac{dw}{r^2 w ^2+1  } =\frac{\pi}{r}<\infty$ and
the point-wise limit of $b_\nu(x)$ is $\alpha'(0)x$, the Lebesgue 
dominated convergence theorem states that
$
\lim_{0^+} D_\nu=\frac{\pi}{|\alpha'(0)|}$. Considering
that 
\begin{equation} \label{eq:bdsquare}
\left|  Q(\mathbf{U})   \right|  \leq C_\theta^1 \| H\|
\end{equation}
 using (\ref{eq:bd112}), there exists a
continuous function 
 $\nu \mapsto
\varepsilon^1(\nu)$ with $\varepsilon^1(0)=0$
such that 
\begin{equation} \label{eq:113}
\left|
\nu \| S_1\|_{L^2(-L-,H)  }^2- \left| \frac{\pi  Q(\mathbf{U})  ^2}{\alpha'(0)}\right|
\right|
\leq C_\theta^1 \varepsilon^1(\nu) \| H\|^2.
\end{equation}

\item The functions 
$\frac{\Dx A_\nu(x)-\Dx A_\nu (0)}{\alpha(x)+i\nu}$ and  $\frac{\Dx B_\nu(x)-\Dx B_\nu (0)}{\alpha(x)+i\nu}$ can  bounded in  
$L^\infty$ uniformly with respect
to $\nu$. 
So
$
\int_{-L}^H |S_2(z)|^2 dz \leq c^2_\theta \| 0\|^2$. 
Estimate 
(\ref{eq:bd74}) yields
\begin{equation} \label{eq:s2}
\nu \| S_2\|_{L^2(-L-,H)  }^2 
\leq C_\theta^2 \nu (1
+| \ln \nu |)^2 \| H\|^2, \qquad C_\theta^2>0.
\end{equation} 

\item The last term $S_3$ is 
$$
|S_3(x)|= 
\left| \int_0^x \frac{\Dx \Dz k^\nu(x,z)}{\alpha(x)+i\nu } 
  U(z)dz \right|
\leq c_3^\theta  \left| \int_0^x |U(z)|dz \right|
$$
since 
the kernel is  bounded (\ref{eq:ddfond})
with respect to $\theta$ and uniformly for $\nu\in [0,1]$.
Inequality (\ref{eq:pol})
implies 
that
$
|S_3(x)|\leq 
c_\theta^3 \left( 1+|\ln\nu|  \right)^2 \|H\|$.
Therefore
 this term is bounded like
\begin{equation} \label{eq:s3}
\nu \| S_3\|_{L^2(-L-,H)  }^2 
\leq
c_\theta^4 
\nu \left(1+|\ln\nu|\right)^4
 \| H\|^2
,
\qquad
c_\theta^4>0.
\end{equation} 
\end{itemize}
We complete the proof adding  the three inequalities (\ref{eq:113}-\ref{eq:s3}).
\end{proof}

To pursue the analysis, 
we begin by rewriting 
 the general form of the
integral equation (\ref{eq:bdeq:bd13}),
showing that the various
singularities of the equation can be recombined
  under a more convenient form.
 This intermediate  result is essential
to obtain all following results.
Indeed the integral equation for $U$ (\ref{eq:bdeq:bd13}) choosing $G=0$ writes 
$$
(\alpha(x)+i\nu )U(x)=
a_0 \Dx A_{\nu} (x)+ b_0 \Dx B_{\nu}(x)+
\int_0^x \Dx \Dz k^{\nu}(x,z) U(z)dz.
$$
Since by construction 
$a_0 \Dx A_{\nu} (0)+ b_0 \Dx B_\nu(0)=Q(\mathbf{U})$ one also has
$$
(\alpha(x)+i\nu )U(x)= a_0 
\left( \Dx A_{\nu} (x)- \Dx A_\nu(0)\right)
+ b_0 \left( \Dx B_{\nu}(x) -\Dx B_\nu(0)\right)
$$
$$
+Q(\mathbf{U})
+
\int_0^x \Dx \Dz k^{\nu}(x,z) U(z)dz
.
$$
But one also has
due to the integral equation for $V$ (\ref{eq:bd11}) choosing $G=H$ 
$$
V(0)=a_0=a_H A_{\nu}(0)+b_H B_{\nu}(0)-\int_0^H
\Dz  k^{\nu}(0,z) U(z)dz.
$$
Basic manipulations yield
$$
a_0=a_H -\int_0^H
\left( \Dz  k^{\nu}(0,z) -\Dz  k^{\nu}(0,0) \right)
U(z)dz
-
i\theta \int_0^H   U(z)dz
$$
because $\Dz  k^{\nu}(0,0)=i\theta$. 
Since the function $ \Dz  k^{\nu}$ is continuous, there exists
a constant $C_4^\theta$ independent of $\nu$ such that 
$$
\left| \Dz  k^{\nu}(x,z) -\Dz  k^{\nu}(x,x)  \right| \leq  C_4^\theta  (z-x)
\leq C_4^\theta z 
\quad \mbox{ for }
0 \leq x \leq z \leq H.
$$
Therefore the integral 
$
\int_0^H
\left|  \Dz  k^{\nu}(0,z) -\Dz  k^{\nu}(0,0) \right|
|U(z)|dz\leq
C_4^\theta \int_0^H z |U(z)|dz 
$
is bounded uniformly with respect to $\nu$ thanks to the 
bound given in (\ref{eq:bd71}). 
We summarize this as 
\begin{equation} \label{eq:bda01}
a_0=\tilde a-i\theta \int_0^H U(z)dz
\end{equation}
where $|\tilde a|\leq C_5^\theta \| H\|$ is bounded uniformly with respect to $\nu$.
Similarly 
\e{
b_0 = b_H-\int_0^H \partial_x \Dz  k^{\nu}(0,z) U(z) dz 
}
and since the function $\partial_x \Dz  k^{\nu}$ is continuous and $\partial_x\mathcal D^\theta_z(0,0)= i \delta(0)$
\begin{equation} \label{eq:bdb01}
b_0=\tilde b- i  \int_0^H \delta (0) U(z)dz
\end{equation}
where $\tilde b$ is also bounded uniformly with respect to $\nu$:
$|\tilde b|\leq C_6^\theta \| H\|$.
The integral equation then gives 
$$
(\alpha(x)+i\nu )U(x)= \tilde a
\left( \Dx A_{\nu} (x)- \Dx A _\nu(0)\right)
+ \tilde b \left( \Dx B_{\nu}(x) -\Dx B_\nu(0)\right)
$$
$$
+Q(\mathbf{U})
-\int_0^H Q(x,z) U(z)dz
+
\int_0^x \Dx \Dz k^{\nu}(x,z) U(z)dz
$$
where the new kernel is
$$
Q(x,z)=\left( \Dx A_{\nu} (x)- \Dx A _\nu(0)\right)i\theta
+ \left( \Dx B_{\nu}(x) -\Dx B_\nu(0)\right) i\delta (0)
$$
$$
=
 \Dx A^{\nu} (x) \Dz B_\nu(0)-
 \Dx B^{\nu}(x) \Dz A_\nu(0)=
 \Dx \Dz k^\nu(x,0)
$$
after evident simplifications.
It is convenient to introduce two 
bounded functions $m^{\theta,\nu}=
\frac{\Dx A^{\nu} (x)- \Dx A_\nu(0)}x
$ and $n^{\theta,\nu}=
\frac{\Dx B^{\nu} (x)- \Dx B_\nu(0)}x$
so that (\ref{eq:bdeq:bd13})  is rewritten as 
\begin{equation} \label{eq:bd100}
(\alpha(x)+i\nu )U(x)= \tilde a m^{\theta,\nu} (x) 
x
+\tilde b  n^{\theta,\nu} (x) 
x
+Q(\mathbf{U})
-\int_x ^H \Dx \Dz k^\nu(x,0) U(z)dz
\end{equation}
$$
+ \int_0^x \left(   \Dx \Dz k^\nu(x,z) 
-\Dx \Dz k^\nu(x,0)
\right)
U(z)dz , \quad \forall x\in [-L,\infty[.
$$

A first property which shows that 
(\ref{eq:bd100}) is less singular that its initial form
(\ref{eq:bdeq:bd13})  is the following lemma which uses the 
pointwise estimate  (\ref{eq:bd71}) on $U$
(so an  important restriction is  nevertheless  that $x>0$).

\begin{lem}\label{lem:eqintU}
The first component $U$ of any element $\mathbf U \in \mathbb X^{\theta,\nu}$ satisfies
\begin{equation} \label{eq:bd101}
(\alpha(x)+i\nu )U(x)= 
p^{\theta,\nu} (x) x 
+
Q(\mathbf{U})-\int_x^H \Dx \Dz k^\nu(x,0) U(z)dz
\end{equation}
where 
\begin{equation} \label{eq:ptm}
\|p^{\theta,\nu}
\|_{L^\infty(0,H)  } \leq C ^\theta \|H\|, \qquad
\forall \nu\in [0,1].
\end{equation}
\end{lem}
\begin{proof}
Let us focus on the second integral in (\ref{eq:bd100}).
Continuity properties with respect to the second variable $z$
imply that  there exists a constant $C_7^\theta$
independent of $\nu$ such that
\begin{equation} \label{eq:bd130}
\left|
\Dx \Dz k^\nu(x,z)-\Dx \Dz k^\nu (x,0)
\right| \leq C_7^\theta z.
\end{equation}
So, for $ x\geq 0$, 
$$\left|
\int_0^x \left(   \Dx \Dz k^\nu(x,z) 
-\Dx \Dz k^\nu(x,0)
\right)
U(z)dz
\right|\leq C_7^\theta
\int_0^x z |U(z)|dz 
\leq C_7^\theta C_\theta \| H\|x
$$
using estimate (\ref{eq:bd71}).
Set 
\begin{equation} \label{eq:ptm2}
p^{\theta,\nu}(x)=\tilde a m^{\theta,\nu}(x) 
+\tilde b  n^{\theta,\nu}(x)
+\frac1x 
\int_0^x \left(   \Dx \Dz k^\nu(x,z) 
-\Dx \Dz k^\nu(x,0)
\right)
U(z)dz 
\end{equation}
which satisfies by construction (\ref{eq:ptm}).
\end{proof}

As a consequence one has 
\begin{prop} \label{prop:4.11}
For all $1\leq p < \infty$, there exists a constant $C_p^\theta$
independent of $\nu$ and which depends continuously on $\theta$
such that
\begin{equation} \label{eq:bd102}
\left\|
 U-
\frac{Q(\mathbf{U})}{\alpha(\cdot)+i\nu   }
\right\|_{ L^p(0,H)} \leq C_p^\theta \| H\|.
\end{equation}
\end{prop}

\begin{proof}
From lemma \ref{lem:eqintU} 
 one has that
$$
U(x)
- \frac{Q(\mathbf{U})}{\alpha(x)+i\nu}=
\frac{x}{\alpha(x)+i\nu}p^{\theta,\nu} (x)
- \frac{
\Dx \Dz k^\nu(x,0) }{\alpha(x)+i\nu}
\int_x^H U(z)dz, 
$$
which turns into 
$$
\left(
U(x)
- \frac{Q(\mathbf{U})}{\alpha(x)+i\nu}\right)
+\frac{\Dx \Dz k^\nu(x,0)}
{\alpha(x)+i\nu}
\int_x^H  \left(
U(z)
- \frac{Q(\mathbf{U})}{\alpha(z)+i\nu}
\right)
dz
$$
\begin{equation}\label{eq:bd105}
=
\frac{x}{\alpha(x)+i\nu}p^{\theta,\nu} (x)
- Q(\mathbf{U})
\frac{\Dx \Dz k^\nu(x,0)}
{\alpha(x)+i\nu}
\int_x^H \frac1{\alpha(z)+i\nu}dz
\end{equation}
By virtue of (\ref{H2}) we notice that
$
\left|
\int_x^H \frac1{\alpha(z)+i\nu}dz
\right|
\leq 
\int_x^H \frac1{|\alpha(z)|}dz
\leq \frac1r  \log(H/x)$. 
Since all powers
of the function $x\mapsto \ln |x| $ are integrable,
 the right-hand side
(\ref{eq:bd105}) 
is naturally bounded in any $L^p$, $1\leq p<\infty$.
Therefore the function
$Z(x)=U(x)
- \frac{Q(\mathbf{U})}{\alpha(x)+i\nu}$ is solution of an integral equation
with a bounded kernel and a  right hand side in $L^p$.
The form of this integral equation is
$$
Z(x)+ \widetilde{K}^{\theta,\nu}(x) \int_x^H Z(z)dz=
b^{\theta,\nu}(x) 
$$ 
with $\left\|\widetilde{K}^{\theta,\nu}(x)\right\|_{L^\infty(0,H)  }
\leq C_8^\theta$ independently of $\nu$. One also uses 
$
\left\| b^{\theta,\nu} \right\|_{L^p(0,H)   }\leq c_p^\theta  \| H\|$
for 
$0\leq \nu \leq 1$:
the key estimate is (\ref{eq:ptm}) which explains why 
the result is restricted to $x>0$ .
Since this is a standard non-singular integral equation, see \cite{Tricomi}, the claim is proved.
\end{proof}

The previous result (\ref{eq:bd102})
shows that some singularities of the 
integral equation can be recombined in a less singular formulation, so that
the dominant part of $U$ is $ \frac1{\alpha(\cdot )+i\nu   }$.
An important restriction of this technique, for the moment,
is that it needs the  a priori estimate (\ref{eq:bd71}) on $U$.
 This explains
why  inequality (\ref{eq:bd102}) is restricted to $x>0$.
By inspection of the structure of the algebra, it appears that 
one has the same kind
of inequalities  on the entire interval by replacing
$U$ directly by the function $ \frac1{\alpha(\cdot )+i\nu   }$
in the integrals.
A preliminary and fundamental 
result in this direction concerns the function
$$
D^{\theta,\nu} (x)=-\frac{ \Dx \Dz k^\nu(x,0) }{\alpha(x) +
i\nu }\int_0^H \frac1{\alpha(z)+i\nu   }dz
+\int_0^x \frac{\Dx \Dz k^\nu(x,z) }{\alpha(x) +i\nu  }
\frac1{\alpha(z)+i\nu   }dz
$$
which is nothing than the integral part
of (\ref{eq:bd100})
where $U$ is replaced by the function $ \frac1{\alpha(\cdot )+i\nu   }$.

\begin{prop}
Let $1\leq p < \infty$. One has
$
\left\|
D^{\theta,\nu}
\right\|_{ L^p(-L,H)} \leq C_p^\theta
$
where the constant depends continuously on $\theta$ and 
does not depend on $\nu$.
\end{prop}

\begin{proof}
Two cases occur.
\begin{itemize}
\item {\bf Assume $0\leq x \leq H$}. The analysis is similar to the one
of proposition \ref{prop:4.11}. One has the same
kind of rearrangement  (\ref{eq:bd100}), that is 
$$
D^{\theta,\nu}
(x)=-\frac{ \Dx \Dz k^\nu(x,0) }{\alpha(x) +
i\nu }\int_x^H \frac1{\alpha(z)+i\nu   }dz
$$
$$
+\int_0^x \frac{\Dx \Dz k^\nu(x,z) - \Dx \Dz k^\nu(x,0) }{\alpha(x) +i\nu  }
\frac1{\alpha(z)+i\nu   }dz.
$$
The first term
is bounded like  $C^\theta  \frac{  | \log x |}r$
 which is in all $L^p$, $p<\infty$.
The second term is immediately bounded using (\ref{eq:bd130}):
indeed
$$
\left|
\int_0^x \frac{\Dx \Dz k^\nu(x,z) - \Dx \Dz k^\nu(x,0) }{\alpha(x) +i\nu  }
\frac1{\alpha(z)+i\nu   }dz
\right|
$$
$$ 
\leq C_7^\theta \frac{1}{\sqrt{ \alpha(x)^2+\nu^2 }  }\int_0^x \frac{z
  }{ \sqrt{ \alpha(z)^2+\nu^2 }  }dz
\leq C_7^\theta \frac1{r^2}.
$$

\item {\bf Assume $-L\leq x \leq 0$}.  The decomposition is slightly 
different and uses some cancellations permitted by the  symmetry
properties of the kernels.
One has
$$
D^{\theta,\nu}
(x)=-\frac{ \Dx \Dz k^\nu(x,0) }{\alpha(x) +
i\nu }\int_{-x} ^H \frac1{\alpha(z)+i\nu   }dz
$$
$$
+\int_0^x 
\frac{\Dx \Dz k^\nu(x,z)  }{\alpha(x) +i\nu }
\frac1{\alpha(z)+i\nu   }dz
-
\frac{ \Dx \Dz k^\nu(x,0) }{\alpha(x) +
i\nu }\int_0^{-x}  \frac1{\alpha(z)+i\nu   }dz
,
$$
which emphasizes the importance of some symmetry properties of the kernels.
Indeed
$$
\int_0^{-x}  \frac1{\alpha(z)+i\nu   }dz
=
-\int_0^{x}  \frac1{\alpha(-w)+i\nu   }dw
$$
$$
=\int_0^{x}  \frac1{\alpha(w)+i\nu   }dw+
\int_0^{x}  \left(
\frac1{-\alpha(-w)-i\nu   }
-
\frac1{\alpha(w)+i\nu   } \right)dw.
$$
Notice that
$$
\frac1{-\alpha(-w)-i\nu   }
-
\frac1{\alpha(w)+i\nu   }=
\frac{\alpha(w)+\alpha(-w)+2i\nu   }{  (\alpha(w)+i\nu)(-\alpha(-w)-i\nu ) }.
$$
So, since $\alpha(0)=0$,
$$
\left| \frac1{-\alpha(-w)-i\nu   }
-
\frac1{\alpha(w)+i\nu   } \right|
\leq \frac{ 2\left\| \alpha  \right\|_{W^{2,\infty}(-L,H) } w^2 +2 \nu }{r^2 w ^2 +\nu^2},
$$
since $\alpha\in W^{2,\infty}(-L,H)$. 
  One can bound
$$
\left| \int_0^{x}  \frac1{-\alpha(-w)-i\nu   }dw-
\int_0^{x}  \frac1{\alpha(w)+i\nu   }dw \right|
\leq 
\frac{\left\| \alpha  \right\|_{W^{2,\infty}(-L,H) }}{r^2} |x| +
\int_0^x \frac{2\nu}{r^2z^2 +\nu^2} dz
$$
$$
\leq 
\frac{\left\| \alpha  \right\|_{W^{2,\infty}(-L,H) }}{r^2} \max(H,L)
+ \int_0^\infty  \frac{2\nu}{r^2z^2 +\nu^2} dz
\leq 
\frac{\left\| \alpha  \right\|_{W^{2,\infty}(-L,H) }}{r^2} \max(H,L)
+\frac{\pi}{r}.
$$
As a consequence $D^{\theta,\nu}$ can be expressed as 
$$
D^{\theta,\nu}
(x)=-\frac{ \Dx \Dz k^\nu(x,0) }{\alpha(x) +
i\nu }\int_{-x} ^H \frac1{\alpha(z)+i\nu   }dz
$$
$$
+\int_0^{x} 
\frac{\Dx \Dz k^\nu(x,z) - \Dx \Dz k^\nu(x,0)}{\alpha(x) +i\nu  }
\frac1{\alpha(z)+i\nu   }dz
+R(x)
$$
with $\left\|R\right\|_\infty(-L,H)\leq C_{10}^\theta$.
The two integrals have the same structure
as for the first case, in particular the interval of integration
is $[-x,H]$ with $0\leq -x$.  So the same result holds.
\end{itemize}
\end{proof}

\begin{prop}
For all $1\leq p < \infty$, there exists a constant $C_p^\theta $
independent of $\nu$ such that
\begin{equation} \label{eq:bd110}
\left\|
 U-
\frac{Q(\mathbf{U})}{\alpha(\cdot)+i\nu   }
\right\|_{ L^p(-L,H)} \leq C_p^\theta \| H\|.
\end{equation}
\end{prop}

\begin{proof}
We start from (\ref{eq:bd100})   written as
$$
U(x)= \frac{Q(\mathbf{U})}{ \alpha(x)+i\nu }
+\frac{x}{\alpha(x)+i\nu} \widetilde{p}^{\theta,\nu}(x)
$$
$$
-\int_0^H \frac{ \Dx \Dz k^\nu(x,0) }{\alpha(x)+i\nu} U(z)dz
+ \int_0^x \frac{ \Dx \Dz k^\nu(x,z) }{\alpha(x)+i\nu} U(z)dz .
$$
Here
$\widetilde{p}^{\theta,\nu}(x)=\tilde a m^{\theta,\nu}(x) 
+\tilde b  n^{\theta,\nu}(x)$, so that 
 $\|\widetilde{p}^{\theta,\nu}\|_{L^\infty(-L,H)}\leq C^\theta \|H \|$
over the whole interval $(-L,H)$.
Notice that $\widetilde{p}^{\theta,\nu}$ is the first
part of ${p}^{\theta,\nu}$ defined in (\ref{eq:ptm2}). 
Setting $u(x)=U(x)- \frac{Q(\mathbf{U})}{ \alpha(x)+i\nu }  $
one gets 
$$
u(x)
- \int_0^x \frac{ \Dx \Dz k^\nu(x,z) }{\alpha(x)+i\nu} u(z)dz 
$$
$$
=\frac{x}{\alpha(x)+i\nu} \widetilde{p}^{\theta,\nu}(x)
-
Q(\mathbf{U}) D^{\theta,\nu}(x)
-
\int_0^H \frac{ \Dx \Dz k^\nu(x,0) }{\alpha(x)+i\nu}
u(z)dz
.
$$
The left-hand side  is an non singular integral operator of the second
 kind 
with a bounded kernel thanks to the fundamental property
(\ref{eq:ddfond}).
The right-hand side is bounded in $L^p$ with a  continuous
dependence with respect to $\|H\|$, see Lemma \ref{lem:eqintU},
estimation (\ref{eq:bdsquare}) and estimation (\ref{eq:bd102}).

\end{proof}

\subsubsection{The second basis function}

We  apply the above material
to 
the second basis
function for which $Q(\mathbf{U}_2^{\theta,\nu})=1$.
The inequality 
(\ref{eq:bd110}) writes
\begin{equation} \label{newa}
\left\|
 U_2^{\theta,\nu}-
\frac{1}{\alpha(\cdot)+i\nu}
\right\|_{ L^p(-L,H)} \leq C_p^\theta 
\left(
\left|  V_2^{\theta,\nu}(H)  \right|+ \left| 
 W_2^{\theta,\nu}(H)  \right| \right), 
\end{equation}
for $ 1 \leq p < \infty$.

\begin{prop}
Assume the transversality condition (\ref{eq:bd63}). There exists a constant
$C^\theta$ independent of $\nu$ and continuous with
respect to $\theta$ such that 
\begin{equation} \label{eq:bd133}
\left|  V_2^{\theta,\nu}(H)  \right|+ \left| 
 W_2^{\theta,\nu}(H)  \right| \leq
C^\theta.
\end{equation}
\end{prop}
\begin{proof}
Indeed, regarding relation \eqref{eq:bd47.6}, \eqref{eq:bd48} 
the pair $(v,w)=(V_2^{\theta,\nu}(H),
 W_2^{\theta,\nu}(H)   )$ is solution of the linear
system
\begin{equation} \label{eq:bd140}
\left\{
\begin{array}{ll}
-v W_1^{\theta,\nu}(H)+ w  V_1^{\theta,\nu}(H)=1, \\
v W_3^{\theta,\nu}(H)- w  V_3^{\theta,\nu}(H)=0.
\end{array}
\right.
\end{equation}
The 
determinant of this linear system is equal to the value of the
function $- \sigma(\theta,\nu)$. So the
transversality condition establishes that
$$
\mbox{det }
\left(
\begin{array}{ll}
 -W_1^{\theta,\nu}(H) &   V_1^{\theta,\nu}(H) \\
 W_3^{\theta,\nu}(H) & - V_3^{\theta,\nu}(H)
\end{array}
\right)
=-\sigma(\theta,\nu)
\neq 0.
$$
Therefore the solution of the linear system 
\begin{equation} \label{eq:fu6}
v=- \frac{ V_3^{\theta,\nu}(H) } { \sigma(\theta,\nu)  }, \quad 
w= - \frac{ W_3^{\theta,\nu}(H) } {\sigma(\theta,\nu)  }
\end{equation}
is bounded uniformly with respect to $\nu$.
\end{proof}

\begin{thm}\label{thm5.1}
Assume  the same transversality condition (\ref{eq:bd63}).
 The second basis function satisfies the following
estimates 
for some $C_p^\theta$ and $C^\theta$ which are continuous with respect
to $\theta$
\begin{equation} \label{eq:bd137}
\left\| U_2^{\theta,\nu}- \frac{1}{\alpha(\cdot)+i\nu   } \right\|_{ L^p(-L,H)} \leq C_p^\theta , \quad 1\leq p < \infty,
\end{equation}
\e{
 \label{eq:bd137.4}
\left\| \mathbf{U}^{\theta,\nu}_2  \right\|_
{ H^1_{\rm loc} ([-L,0)\cup (0,H])}
\leq C^\theta.
}
\end{thm}

\begin{proof}
The first estimate is a straightforward consequence of \eqref{newa}, \eqref{eq:bd133} . 
The use of the integral representations 
(\ref{eq:bd11}-\ref{eq:bd12}) shows that,
\begin{equation} \label{eq:bd137.2}
\left\| V_2^{\theta,\nu}\right\|_
{ L^\infty_{\rm loc} ([-L,0)\cup (0,H])}
+
\left\| W_2^{\theta,\nu}\right\|_
{ L^\infty_{\rm loc} ([-L,0)\cup (0,H])}
 \leq C^\theta 
\end{equation}
for some 
$C^\theta$.
Then the second equation of (\ref{sys0:hatmu}) shows that
one has the same bound for $U_2^{\theta,\nu}$
\begin{equation} \label{eq:bd137.3}
\left\| U_2^{\theta,\nu}\right\|_
{ L^\infty_{\rm loc} ([-L,0)\cup (0,H])}
 \leq C^\theta .
\end{equation}
The bound on the derivatives follows from (\ref{sys0:hatmu}) 
\end{proof}

\begin{rmk} 
Let us set $H'=-L$.
From (\ref{eq:bd137.2}) one gets that $\|H'\|$ is bounded uniformly also,
therefore (\ref{eq:bd71}) can be generalized for $x<0$ (resp.  $H'$)
instead of $x>0$  (resp. $H$).
In summary one has 
for a constant $K^\theta$ that can be further specified: 
$
\left| U_2^{\theta,\nu}(x)\right| \leq \frac{K^\theta}{\sqrt{r^2 x ^2 + \nu^2}}$
for $x\in (-L,H)$.
\end{rmk}
We now pass to the limit $\nu\rightarrow 0^\pm$.

\begin{prop}
Assume the same transversality condition (\ref{eq:bd63}).
 The second basis function 
admits a limit in the sense of distribution
for $\nu=0^\pm$ as follows:
$$
\mathbf{U}_2^{\theta,\nu} \rightarrow 
\mathbf{U}_2^{\theta,\pm}
=
\left(
P.V. \frac1{\alpha(x)} \pm \frac{i \pi}{\alpha'(0)} \delta_D
+ u_2^{\theta,\pm}, \;  
v_2^{\theta,\pm}, \; w_2^{\theta,\pm}
\right)
$$
where $u_2^{\theta,\pm},
v_2^{\theta,\pm}, w_2^{\theta,\pm}\in L^2(-L,\infty)$
and $\delta_D$ is the Dirac mass at the origin.
\end{prop}

\begin{rmk}
The limits $\mathbf{U}_2^{\theta,\pm}$
are solutions of (\ref{sys0:hatmu}) in the sense of distribution.
they will be called
the singular solutions.
\end{rmk}

\begin{proof}

We consider firstly the case $\nu\downarrow 0$.
Some parts of the proof are already evident,
essentially
for quantities which  are regular enough ($V_2^{\theta,\nu}$ and
$W_2^{\theta,\nu}$) or for regions where all functions are regular 
(typically $x>0$).
Therefore the whole point is to pass to the limit in the singular
part of the solution $U_2^{\theta,\nu}$.
We will make wide use of the equivalence between the integral formulation
of proposition
\ref{prop:bd1} and the differential formulation
(\ref{sys0:hatmu}).

$\bullet$ {\bf Passing to the weak limit:}
By continuity of the first 
basis function with respect to $\nu$, one can pass to the limit
concerning  $(V_2^{\theta,\nu}(H), W_2^{\theta,\nu}(H)   )$.
One gets
that  $(v,w)=(V_2^{\theta,0^+}(H),
 W_2^{\theta,0^+}(H)   )$
is the unique solution of the linear system
\begin{equation} \label{eq:bd141}
\left\{
\begin{array}{ll}
-v W_1^{\theta}(H) +w  V_1^{\theta}(H)=1, \\
v W_3^{\theta}(H)- w  V_3^{\theta}(H)=0,
\end{array}
\right.
\end{equation}
where the coefficients are defined in terms of the first basis
function for $\nu=0$.
By continuity away from the singularity at $x=0$, one has that
$
\mathbf{U}_2^{\theta,\nu}
\rightarrow
\mathbf{U}_2^{\theta}
$
in 
$L^\infty(\epsilon, H)$ for all $\epsilon>0$.
Using (\ref{eq:bd137}) it is clear
that $U^{\theta,\nu}-\frac1{\alpha(\cdot)+i\nu}
$ is bounded in $L^2(-L,H)$ 
uniformly   with respect to $\nu$. Therefore
there exists a limit function denoted as
$u_2^{\theta,0^+}$ such that for a subsequence:
$
U_2^{\theta,\nu}-\frac1{\alpha(\cdot)+i\nu}
\rightarrow_{\rm weak} 
u_2^{\theta,0^+}
$ in $L^2(-L,H)$. 
Moreover the first derivative of 
$U_2^{\theta,\nu}$ is bounded in 
$L^2(-L,-\epsilon)$ by virtue of (\ref{eq:bd137.4}). 
Therefore  
$
U_2^{\theta,\nu}
\rightarrow_{\rm strong}
\frac1{\alpha(\cdot)}+ 
u_2^{\theta,0^+}
$ in $ L^2(-L,-\epsilon)$
at least for a subsequence.
Considering the integral relations (\ref{eq:bd11}-\ref{eq:bd12}),
these subsequences are such that 
\begin{equation} \label{new1}
V_2^{\theta,\nu}(x)
\rightarrow
v_2^{\theta,0^+}(x),
\end{equation}
and
\begin{equation}\label{new2}
W_2^{\theta,\nu}(x)
\rightarrow
w_2^{\theta,0^+}(x),
\end{equation}
with  the convergence  uniform in compact sets of $(-L, H)\setminus\{0\} $.
The limits in  (\ref{new1}), (\ref{new2}) also hold in the strong topology of  $L^2(-L,H)$.
To be more complete we detail hereafter some formulas
which can be derived for these functions.
Let us consider  $0 < \epsilon$ a real number, a priori small, so that
   $\alpha(x)$ is invertible on the interval
$[-\epsilon,\epsilon]$.
 We define
$\beta(z):=
1/ \alpha'(\alpha^{-1}(z))$ with $\alpha^{-1}$ the inverse function of $ \alpha$. 
Let us consider the principal branch of the complex logarithm.
One can check that
$$
v_2^{\theta, 0^+}:= a_H \, A_0(x) +b_H\, B_0(x)+\int_H^x
 {\mathcal D}^{\theta}_x (k^0(x,z)- k^0(x,0)) \left( \frac{1}{\alpha(z)}
 +u_2^{\theta,+}(z)\right)+
 $$
 $$
\int_H^x\,
 {\mathcal D}^{\theta}_x k^0(x,0)  u_2^{\theta,+}(z)+ \tilde{v}(x),
 $$
 where the function $\tilde{v}$ is 
 $$
\tilde{v}(x):=
\int_H^x \,{\mathcal D}^{\theta}_x k^0(x,0) \,\frac{1}{\alpha(z)}\, d z,\quad  \mbox{for } x  > 0,
$$
and on the other side of the singularity 
\begin{equation} \label{eq:vvtilde}
\tilde{v}(x):= \mathcal D^{\theta}_x k^0(x,0)\left[ \int_x^{ -\epsilon}\, \frac{1}{\alpha(z)}\, dz + \ln \alpha(\epsilon)  \,\beta(\epsilon)- \right.
\end{equation}
$$\left.
\ln\alpha(-\epsilon) \,\beta(-\epsilon)+
\int_{\alpha(\epsilon)}^{\alpha(-\epsilon)}\, \ln(z)\, \beta'(z)\, dz\right]+  \int_\epsilon^H \,{\mathcal D}^{\theta}_x k^0(x,0) \,\frac{1}{\alpha(z)}\,
 d z,\quad   \mbox{for }x  < 0.
$$
Similarly one has 
$$
w_2^{\theta, 0^+}:= a_H \, A'_0(x) +b_H\, B'_0(x)+\int_H^x
\partial_x {\mathcal D}^{\theta}_x (k^0(x,z)- k^0(x,0)) \left( \frac{1}{\alpha(z)} +u_2^{\theta,+}(z)\right)
 $$
 $$
+
\int_H^x\,
 \partial_x {\mathcal D}^{\theta}_x k^0(x,0)  u_2^{\theta,+}(z)+ \tilde{w}(x),
 $$
 with on one side 
 $$
\tilde{w}(x):=
\int_H^x \,\partial_x{\mathcal D}^{\theta}_x k^0(x,0) \,\frac{1}{\alpha(z)}\, d z,\quad  \mbox{for } x  > 0,
$$
and on the other side of the singularity
\begin{equation} \label{eq:wwtilde}
\tilde{w}(x):= \partial_x{\mathcal D}^{\theta}_x k^0(x,0) \left[ \int_x^{ -\epsilon}\, \frac{1}{\alpha(z)}\, dz + \ln \alpha(\epsilon)  \,\beta(\epsilon)- \right.
\end{equation}
$$\left.
\ln\alpha(-\epsilon) \,\beta(-\epsilon) +
\int_{\alpha(\epsilon)}^{\alpha(-\epsilon)}\, \ln(z)\, \beta'(z)\, dz\right]
+  \int_\epsilon^H \,\partial_x{\mathcal D}^{\theta}_x k^0(x,0) \,\frac{1}{\alpha(z)}\, d z,\quad  \mbox{for } x  < 0.
$$
These weak or strong limits are naturally weak solutions
of the initial system (\ref{sys0:hatmu}): 
denoting for simplicity
$
(u_2,v_2,w_2)=
(u_2^{\theta,0+}, v_2^{\theta,0+},    w_2^{\theta,0+} )$, 
these functions are solutions of 
\begin{equation} \label{eq:bd145}
\left\{
\begin{array}{lll}
\displaystyle 
\int  w_2 \varphi_1 dx 
+ i\theta 
\; P.V.  \int \left( \frac1{\alpha} +u_2 \right)  \varphi_1  dx-
 \frac{\theta \pi}{\alpha'(0)  } \varphi_1(0)
+
\int v_2  \varphi_1'  dx =0 , \\
\displaystyle 
i\theta \int w_2 \varphi_2 dx  -\int (\alpha  u_2 +1)\varphi_2
dx
 -i\int \delta   v_2 \varphi_2 dx =0,\\ 
\displaystyle
\int w_2 \varphi_3 ' dx
 + i \; P.V. \int \delta  \left( \frac1{\alpha}+
u_2 \right)\varphi_3 dx -
  \frac{ \delta(0) \pi}{\alpha'(0)  } \varphi_3(0)  
\\
~ \hspace{7.cm}  - \int \alpha v_2 \varphi_3 dx=0,
\end{array}
\right.
\end{equation}
for any sufficiently smooth test
functions with compact support, for example
 $(\varphi_1,\varphi_2,\varphi_3)\in {\cal C}^1_0(-L,H)$.
To pass to the limit we have used 
that in distribution sense, $\lim_{\nu \rightarrow 0^+} \frac{1}{\alpha(x)+i \nu}= P.V \frac{1}{\alpha(x)}+i \pi  \frac{1}{\alpha'(0)} \delta_D$.
 The signs of 
$-
 \frac{\theta \pi}{\alpha'(0)  } \varphi_1(0)$
and $-
  \frac{ \delta(0) \pi}{\alpha'(0)  } \varphi_3(0)  $
are compatible with the fact the limit is for positive $\nu$. The principal value  is defined as:
 $$ 
 P.V. \int \frac{1}{\alpha(x)}\, \varphi (x)\, dx:= \lim
_{\epsilon \downarrow 0} \left( \int_{-L}^{\rho(-\epsilon)}       \frac{1}{\alpha(x)}\, \varphi (x)     +\int_{\rho(\epsilon)}^H       \frac{1}{\alpha(x)}\, \varphi (x)\right)\, dx,
$$
where $\alpha(\rho(\mp \epsilon))= \pm \epsilon.$

$\bullet$ {\bf Uniqueness of the weak limit:} If there is another triplet
$(\widetilde{u_2},\widetilde{v_2},\widetilde{w_2})$
solution of the same weak formulation (\ref{eq:bd145}),  
then the difference 
$
(\widehat{u_2},\widehat{v_2},\widehat{w_2})=
(\widetilde{u_2}-u_2,\widetilde{v_2}-v_2,\widetilde{w_2}-w_2)
$
satisfies
\begin{equation} \label{eq:bd146}
\left\{
\begin{array}{lll}
\displaystyle 
\int  \widehat{ w_2 }\varphi_1 dx 
+ i\theta  \int \widehat{ u_2}
  \varphi_1  dx
+
\int \widehat{ v_2 } \varphi_1'  dx =0 , \\
\displaystyle 
i\theta \int \widehat{ w_2 }
\varphi_2 dx  -\int \alpha  \widehat{ u_2 }\varphi_2
dx
 -i\int \delta(x)  \widehat{  v_2 }\varphi_2 dx =0,\\ 
\displaystyle
\int \widehat{ w_2 }\varphi_3 ' dx
 + i \int \delta  
\widehat{ u_2 } \varphi_3 dx - \int \alpha(x) \widehat{ v_2 } \varphi_3\,dx=0,
\end{array}
\right.
\end{equation}
Because the limit is strong in $L^\infty_loc(]0,H[)$, 
$(\widehat{u_2},\widehat{v_2},\widehat{w_2})=(0,0,0)$ for $x>0$.
For $x<0$, we deduce from
(\ref{eq:bd146}) that  
 $(\widehat{u_2},\widehat{v_2},\widehat{w_2})$
is a solution of the  X-mode equations.
Therefore
these functions can be expressed as a linear combination
of the first and second basis functions for $x<0$.
Since $\widehat{u_2}\in L^2(-L,0)$ is non singular,  
only the first basis function is involved that is 
$$
(\widehat{u_2},\widehat{v_2},\widehat{w_2})
= \lambda
\left(U_1^{\theta}, V_1^{\theta}, W_1^{\theta}   \right)
\qquad x<0.
$$
From (\ref{eq:bd146}) we get
for example
$$
\displaystyle
\int_{-L}^0 \widehat{ w_2 }\varphi_3 ' dx
 + i \int_{-L}^0 \delta  
\widehat{ u_2 } \varphi_3 dx - \int_{-L}^0 \alpha(x) \widehat{ v_2 } \varphi_3\,dx=0
$$
where $\varphi_3(-L)=0$ and  
 $\varphi_3(0)$ is arbitrary.
We integrate by parts
$$
\displaystyle
\int_{-L}^0 \left( - \widehat{ w_2 }'+
 i   \delta  
\widehat{ u_2 } -\alpha \widehat{v_2} \right) \varphi_3 dx +
\widehat{\omega_2}(0)\varphi_3(0)=0.
$$
Since $
(\widehat{u_2},\widehat{v_2},\widehat{w_2})$ is a non singular
solution of the X-mode equations, one has that
$ - \widehat{ w_2 }'+
 i   \delta  
\widehat{ u_2 } -\alpha \widehat{v_2} =0$.
Finally 
$
\widehat{\omega_2}(0)\varphi_3(0)= 0$. 
Since we can take $\varphi_3(0)\neq 0$, it follows that $0=\widehat{\omega_2}(0)=\lambda
W_1^{\theta}(0)$. Considering the normalization
(\ref{H17}) one gets that $\lambda=0$. 
Therefore
 $
(\widehat{u_2},\widehat{v_2},\widehat{w_2})=(0,0,0)$.
It means that the weak limit is unique: all the sequence
tends to the same weak limit.

$\bullet$ {\bf Regularity:} By Theorem \ref{thm5.1}
the limit  belongs to
$H^1\left( \left[-L, -\epsilon\right]\cup
\left[\epsilon, \infty\right) \right)^3$
for all $\epsilon>0$. 

$\bullet$ {\bf Limit $\nu\uparrow 0$:}
The  sign
 of the Dirac mass is changed in the final result
of the proposition since 
$\lim_{\nu \rightarrow 0^-} \frac{1}{\alpha(x)+i \nu}= P.V \frac{1}{\alpha(x)}
-i \pi  \frac{1}{\alpha'(0)} \delta_D$.

\end{proof}

\section{The limit spaces $\mathbb{X}^{\theta,\pm}$} \label{sec:6}

We can now define the limit spaces in which the limit basis functions
live.

\subsection{The space $\mathbb{X}^{\theta,+}$  }
Passing to the limit $\nu\rightarrow 0^+$,
the limit space 
$
\mathbb{X}^{\theta,+}
$
is
\begin{equation} \label{eq:bd150}
\mathbb{X}^{\theta,+}
=\mbox{Span}\left\{ \mathbf{U}_1^{\theta},
 \mathbf{U}_2^{\theta,+}
\right\}\subset
 H^1_{\rm loc}  \left(\left(-L,\infty\right)\setminus\{0\}   \right).
\end{equation}

\subsection{The space $\mathbb{X}^{\theta,-}$  } \label{sec:x-}
 
It is of course possible  do all the analysis
with negative $\nu<0$ and to study the limit
$\nu\rightarrow 0^-$.
The first basis function is exactly the same.
The second basis function  is chosen
exponentially decreasing at infinity and such that
$$
i\nu U_2^{\theta,\nu}=1 \qquad \nu<0.
$$ 
The generalization of the preliminary result (\ref{eq:bd137}) is 
straightforward
\begin{equation} \label{eq:bd157}
\left\|
 U_2^{\theta,\nu}-
\frac{1}{\alpha(\cdot)+i\nu   }
\right\|_{ L^p(-L,H)} \leq C_p^\theta, 
\qquad -1\leq \nu < 0. 
\end{equation}
Passing to the limit $\nu\rightarrow 0^-$, it defines  the
limit  space
$\mathbb{X}^{\theta,-}$
\begin{equation} \label{eq:bd151}
\mathbb{X}^{\theta,-} 
=
\mbox{Span}\left\{ \mathbf{U}_1^{\theta},
 \mathbf{U}_2^{\theta,-}
\right\}
\subset   H^1_{\rm loc}  \left( \left(-L,\infty\right)\setminus\{0\}   \right).
\end{equation}

\subsection{Comparison of the limits}

The first basis function  $ U_1^{\theta}$ 
is independent of the sign
and belongs
to $\mathbb{X}^{\theta,+}
\cap \mathbb{X}^{\theta,-}$.
Since the limit equation and the normalization at $x=H$ are the same, we readily observe that
the limits of the second basis functions are identical 
for $0<x$ 
\begin{equation} \label{eq:bd172}
 U_2^{\theta,+}(x)=  U_2^{\theta,-}(x) 
 \qquad 0<x.
\end{equation}

So the main point  is to determine  the difference
between  the limit of the two singular  functions for $x<0$.
A first remark is that 
$ U_2^{\theta,+}$,   $U_2^{\theta,-}$ and $ U_1^{\theta}$ are three
solutions of the same problem for $x<0$. Since we know
that the dimension of the space   of solution is two,
these functions are necessarily  linearly dependent.

\begin{prop}
One has 
\begin{equation} \label{eq:bd170}
 U_2^{\theta,+}(x)- U_2^{\theta,-}(x)= \frac{-2 i \pi}{\alpha'(0)}   U_1^{\theta}(x)
 \qquad x<0.
\end{equation}
\end{prop}

\begin{proof}
We notice that the Wronskian relations 
(\ref{eq:bd48}) are  the same
at the limit $\nu=0^\pm$.
By subtraction
$$
V_1^{\theta}(x)\left( W_2^{\theta,+}(x)- W_2^{\theta,-}(x)\right)
  -
W_1^{\theta}(x)\left(   V_2^{\theta,\nu}(x)
-
V_2^{\theta,-}(x)
 \right) =0.
$$
It show that the difference is proportional
to the first basis function
\begin{equation} \label{eq:diff}
U_2^{\theta,+}(x)- U_2^{\theta,-}(x)= \gamma   U_1^{\theta}(x)
 \qquad x<0.
\end{equation}
 It remains to determine $\gamma$.
We already now that the limit $\nu\rightarrow 0^+$ can be characterized
by (\ref{eq:bd145}).
The third equation writes
$$
\int w_2^+ \varphi_3 ' dx
 + i \; P.V. \int \delta  \left( \frac1{\alpha(x)}+
u_2^+ \right)\varphi_3 dx -
  \frac{ \delta(0) \pi}{\alpha'(0)  } \varphi_3(0)  
  - \int \alpha(x) v_2^+  \varphi_3 dx=0
$$
where $(u_2^+,v_2^+,w_2^+)$
refers to the non singular part of the limit $\nu\rightarrow 0^+$.
The equivalent equation for 
the non singular part $(u_2^-,v_2^-,w_2^-)$
of the limit $\nu\rightarrow 0^-$
is
$$
\int w_2^- \varphi_3 ' dx
 + i \; P.V. \int \delta  \left( \frac1{\alpha(x)}+
u_2^- \right)\varphi_3 dx +
  \frac{ \delta(0) \pi}{\alpha'(0)  } \varphi_3(0)  
  - \int \alpha(x) v_2^-  \varphi_3 dx=0.
$$
By subtraction, one gets
$$
\int (w_2^+-w_2^-) \varphi_3 ' dx
 + i  \int \delta 
(u_2^+-u_2^-)\varphi_3 dx 
-
  \frac{ 2 \delta(0) \pi}{\alpha'(0)  } \varphi_3(0)  
  - \int \alpha(x) (v_2^+-v_2^-) \varphi_3 dx=0.
$$
Due to (\ref{eq:bd172}) these
 differences vanish for $x>0$.
We get
$$
\int_{-L}^0 (w_2^+-w_2^-) \varphi_3 ' dx
 + i  \int_{-L}^0 \delta 
(u_2^+-u_2^-)\varphi_3 dx 
-
  \frac{ 2 \delta(0) \pi}{\alpha'(0)  } \varphi_3(0)  
  - \int_{-L}^0\alpha(x) (v_2^+-v_2^-) \varphi_3 dx=0
$$
where $\varphi_3$ is a smooth 
test function that vanishes at $-L$.
Integration by part yields
$$
\int_{-L}^0 
\left(
- (w_2^+-w_2^-)'+ i \delta  (u_2^+-u_2^-)
-\alpha(x) (v_2^+-v_2^-)
\right)
\varphi_3  dx
$$
$$
-
  \frac{ 2 \delta(0) \pi}{\alpha'(0)  } \varphi_3(0)  
+(w_2^+-w_2^-) (0)\varphi_3(0)=0.
$$
Due to (\ref{eq:diff}) 
one has that
$
- (w_2^+-w_2^-)'+ i \delta  (u_2^+-u_2^-)
-\alpha(x) (v_2^+-v_2^-)=0 $ for 
$x<0$. 
Since $\varphi_3(0)$ is arbitrary, it means that
$
w_2^+(0)-w_2^-(0)= \frac{ 2 \delta(0) \pi}{\alpha'(0)  }$.
We obtain 
$
\gamma  W_1^\theta(0)=  \frac{ 2 \delta(0) \pi}{\alpha'(0)  }$,
that is
$
i \delta(0) \gamma = \frac{ 2 \delta(0) \pi}{\alpha'(0)  }$. 
Therefore $\gamma=\frac{-2 i \pi}{\alpha'(0) }$.
The claim is proved.
\end{proof}

\section{Proof of the main theorem
\label{sec:main}
}

All the information about the first and second basis
functions is now used to construct the
solution of the system (\ref{sys0:hatmu}) with the boundary
condition (\ref{eq:bc}). 
The function $g$ depends only of the vertical variable $y$.
Under convenient condition $g$ admits the Fourier representation
\begin{equation} \label{eq:fu}
g(y)=\frac1{2\pi}
\int_{\mathbb{R}} \widehat{g}(\theta)e^{i\theta y   }d\theta,
\end{equation}
see $(7.1.4)$ in \cite{horm} for this convention.
We first consider a small but non zero regularization parameter
 $\nu>0$.
For the sake of simplicity we will assume that the transversality
condition is satisfied for all $\theta$ in the support of 
$\widehat{g}$
\begin{equation}  \tag{H6} \label{H6}
\left| \sigma(\theta) \right| \geq c> 0 \quad \forall \theta \in \mbox{supp}
\left(\widehat{g}\right).
\end{equation}
It is just a convenient 
uniform version of the point-wise transversality condition (\ref{eq:bd63}).

\subsection{One Fourier mode} \label{ss1}

For one Fourier mode, one needs to consider
the solution of (\ref{sys0:hatmu}) with
boundary condition 
$$
\widehat{W}^\nu(-L)+i  \sgn(\nu)\lambda \widehat{V}^\nu(-L)= \widehat{g}.
$$
Since we add of course that the solution must decrease (exponentially)
at $x\approx \infty$  to guarantee that no energy comes from infinity,
the solution is proportional to the second basis
function. That is there is a coefficient
$\gamma^{\theta,\nu}$ such that 
$
\widehat{\mathbf{U}}^\nu
=   
\gamma^{\theta,\nu}
\mathbf{U}_2^{\theta,\nu}
$. 
The coefficient satisfies the equation
$$
\gamma^{\theta,\nu} \left( W_2^{\theta,\nu}(-L)+i \sgn(\nu)\lambda V_2^{\theta,\nu}(-L) \right)=\widehat{g}(\theta)
$$
that is
$
 \gamma^{\theta,\nu} = \frac{ \widehat{g}(\theta)}{   \tau^{\theta,\nu}}  
$
from which it is clear that we must study the coefficient/function
\begin{equation} \label{eq:fu2}
\tau^{\theta,\nu}=W_2^{\theta,\nu}(-L)+i \sgn(\nu)\lambda V_2^{\theta,\nu}(-L).
\end{equation}

\begin{prop} \label{prop:pp1}
Assume (\ref{H6}).
For every compact 
set $S \subset \mathbb R$, there exists $\epsilon>0$, $\tau^+$   and $\tau_->0$ such that 
$ \tau^-\leq \left| \tau^{\theta,\nu} \right| \leq  \tau^+$
for $  0 <\nu\leq \epsilon$ and $\theta\in S$. 
\end{prop}
\begin{proof}
The upper bound is a direct consequence of 
(\ref{eq:bd137.2}). 
To prove the lower bound, a useful result 
is the formula which comes from (\ref{eq:bd45})
$$ 
\mbox{Im}\left( W_2^{\theta,\nu}(-L) \overline{V^{\theta,\nu}_2(-L)} \right)
 \geq \nu \int_{-L}^\infty \left|U^{\theta,\nu}_2(x)   \right|^2 dx
$$
Combining with (\ref{eq:bd87}) and 
$Q\left( \mathbf{U}^{\theta,\nu}  \right)=1$ (by construction),
  it yields 
$
\mbox{Im}\left( W_2^{\theta,\nu}(-L) \overline{V^{\theta,\nu}_2(-L)} \right)
\geq \tau_->0$. 
Plugging the definition of $\tau^{\theta,\nu}$ inside this inequality, one gets
$$
\mbox{Im}\left( \tau^{\theta,\nu  } \overline{V^{\theta,\nu}_2(-L)} \right)
\geq \tau_- + \sgn(\nu)\lambda \left|V^{\theta,\nu}(-L)  \right|^2\geq \tau_->0.
$$
Therefore
$
\left|
V_2^{\theta,\nu}(-L)
\right|
\times 
\left|\tau(\theta,\nu)   \right|
\geq  \tau_-  $. 
The  $L^\infty$ bounds (\ref{eq:bd137.2}) shows that  there exists $C>0$ such that
$
C \left|\tau(\theta,\nu)   \right|
\geq  \tau_-$. 
\end{proof}

By \eqref{new1}, \eqref{new2},
\begin{equation} \label{eq:fu2b}
\tau^{\theta, +} := W_2^{\theta,0^+}(-L)+i  \sgn(\nu)\lambda V_2^{\theta,0^+}(-L)= \lim_{\nu  \rightarrow 0^+}
\tau^{\theta,\nu}.
\end{equation}

\begin{prop}
Assume (\ref{H6}).
For every compact set $S \subset \mathbb R$, there exists $\epsilon>0$, $\tau^+$   and $\tau_->0$ such that 
$ \tau^-\leq \left| \tau^{\theta,+} \right| \leq  \tau^+$
for $  0 <  \epsilon$ and $\theta\in S$. 
\end{prop}
\begin{proof} This is immediate from Proposition  \ref{prop:pp1} and \eqref{eq:fu2b}.
\end{proof}

\subsection{Fourier representation of the solution}

The solution of (\ref{sysmu}) with the boundary condition
(\ref{eq:bc}) is given by the inverse Fourier formula
\begin{equation} \label{eq:fu10}
\left(
\begin{array}{c}
E_x^\nu \\
E_y^\nu \\
W^\nu
\end{array}
\right)
(x,y)=\frac1{2\pi}
\int_{\mathbb{R}}
\frac{ \widehat{g}(\theta)}{   \tau^{\theta,\nu}}
\mathbf{U}_2^{\theta,\nu}(x) e^{i\theta y  }d\theta
\end{equation}
 where we assume that $ g \in L^2(\mathbb{R})$ and that $ \widehat{g}$ has compact support.
Since by Theorem  \ref{thm5.1} 
 $\left\| u^{\theta,+} \right\| \leq  C^\theta_2$, $\left\| v^{\theta,+} \right\| \leq  C^\theta_2,$$\left\| w^{\theta,+} \right\| \leq  C^\theta_2$ with $C^\theta_2$ a continuous function of $\theta$,  and  considering that 
$ \tau^{\theta,\nu}$  converges to  $\tau^{\theta,+}$
there is sufficient regularity to pass to the limit in (\ref{eq:fu10}). 
One gets
(\ref{eq:fu11}). The value of the resonant heating
(\ref{eq:limheat}) is obtained by passing to the limit
in the quadratic energy 
$$
{\cal Q}^+=
\lim_{\nu\rightarrow 0^+}
{\cal Q}(\nu)=
\omega\varepsilon_0
\lim_{\nu\rightarrow 0^+}
\mbox{ Im}\left(
\int_\Omega  
\left(\mathbf{E}, \underline{\underline{\varepsilon}}(\nu) \mathbf{E}   \right)
\right),
$$
that is with (\ref{eq:Qphys})-(\ref{eq:espsimp}) 
$$
{\cal Q}^+=
\omega\varepsilon_0
\lim_{\nu\rightarrow 0^+}
\nu \int |E_x^\nu(x,y)|^2 dx dy
=\frac{\omega\varepsilon_0}{2}
\int_{\mathbb{R}}  \frac{  \left|
\widehat{g}(\theta) \right|^2 }{ \left|
\alpha'(0)
\right|
 \left|  \tau^{\theta,+} \right|^2 }
d\theta>0.
$$
We obtain the result with the simplification
 $\omega=\varepsilon_0=1$.

\begin{rmk} 
Observe that  the singular solutions $ {\bf U}^{\theta, \pm}_2$ are the unique solutions of the following initial value problem:
Find a triplet 
$(u^{\theta,\pm}_2, v^{\theta,\pm}_2, w^{\theta,\pm}_2) \in L^{2}(-L, \infty)^3$
which satisfies  the constraints
$v^{\theta,\pm}_2(H)= V^{\theta,0}_3(H)$, $w^{\theta,\pm}_2(H)= W^{\theta,0}_3(H)$, and
$$
\left\{
\begin{array}{clc}
w_2^{\theta,\pm}-\frac{d}{dx}v^{\theta,\pm}_2   +i \theta u_2^{\theta,\pm}
  &=&-i \theta P.V.  \frac{1}{\alpha(x)} \pm \frac{\theta \pi}{\alpha'(0)}
 \delta_D,
  \\
i\theta  w_2^{\theta,\pm} -\alpha(x)  u^{\theta,\pm}_2 
 -i \delta(x)   v^{\theta,\pm}_2  &=
&1,
 \\
-\frac{d}{dx} w^{\theta,\pm}_2 +i \delta(x)
u^{\theta,\pm}_2    - \alpha(x) v^{\theta,\pm}_2 
&=&- iP.V.  \frac{\delta(x)}{\alpha(x)}\pm  \frac{\delta(0) \pi}{\alpha'(0)} \delta_D.
\end{array}
\right.
$$
We prove that this problem has an unique solution by the argument given to prove the uniqueness of the weak limits. For this purpose
observe that 
$$
 \left(\frac{1}{\alpha(x)}+ u_2^{\theta,\pm}, v_2^{\theta,\pm}, w_2^{\theta,\pm} \right)(x)= \left( U_3^{\theta,0}, V_3^{\theta,0}, W_3^{\theta,0}\right)(x)
\mbox{ for }x>0.
$$
We observe the similarity with the standard limiting absorption principle in
scattering theory. In scattering theory the solutions obtained by the
 limiting absorption principle are characterized as the unique solutions that
satisfy the radiation condition, i.e., they are uniquely determined by the
behavior at infinity.
Here, the singular solutions are uniquely determined 
by their behavior at $+\infty$ and by their singular part
$
P.V. \frac1{\alpha(x)} \pm \frac{i \pi}{\alpha'(0)} \delta_D
$
Note that it is natural that we have to specify
the singularity at
$ x = 0$ because our equations are degenerate at $x = 0$.
We think this principle could be used
for  practical computations.  It is however a little
more subtle since a boundary condition 
at finite distance $x=-L$ must be prescribed.
That is the singular part is itself
dependent on the boundary condition where the energy comes
in the system. Mathematically it corresponds to the 
coefficient $\tau^{\theta,+}$ in the representation formula
(\ref{eq:fu11}).
\end{rmk}

\end{document}